\documentclass[12pt]{amsart}

\title{Zigzag Persistence}
\author{Gunnar Carlsson \and Vin de~Silva}
\date{\today.}
\thanks{This work has been supported by DARPA grants HR0011-05-1-0007 and HR0011-07-1-0002.}
\usepackage{fullpage}
\usepackage{amsmath}
\usepackage{amsthm}
\usepackage{amscd}
\usepackage{amssymb,latexsym}
\usepackage{amsfonts}
\usepackage{eucal}
\usepackage{graphicx}
\usepackage{color}
\usepackage[cmtip,arrow]{xy}
\usepackage{pb-diagram,pb-xy}

\numberwithin{equation}{section}

\theoremstyle{plain}
\newtheorem{Theorem}[equation]{Theorem}
\newtheorem{Lemma}[equation]{Lemma}
\newtheorem{Proposition}[equation]{Proposition}
\newtheorem{Corollary}[equation]{Corollary}
\newtheorem{Addendum}[equation]{Addendum}

\newtheorem*{Moral}{Moral}

\theoremstyle{definition}
\newtheorem{Definition}[equation]{Definition}
\newtheorem{Example}[equation]{Example}
\newtheorem{Caution}[equation]{Caution}
\newtheorem{Algorithm}[equation]{Algorithm}

\theoremstyle{remark}
\newtheorem*{Remark}{Remark}
\newtheorem*{Example*}{Example}

\newcommand{\half}{\frac{1}{2}}

\newcommand{\tm}[1]{{\mathbb{#1}}}
\newcommand{\Uu}{\tm{U}}
\newcommand{\Vv}{\tm{V}}
\newcommand{\Ww}{\tm{W}}
\newcommand{\Xx}{\tm{X}}
\newcommand{\Ii}{\tm{I}}

\newcommand{\Ff}{\mathbb{F}}
\newcommand{\fv}[1]{{\mathcal{#1}}}
\newcommand{\Ll}{\fv{L}}
\newcommand{\Rr}{\fv{R}}
\newcommand{\Ss}{\fv{S}}
\newcommand{\Tt}{\fv{T}}
\newcommand{\Jj}{\fv{J}}
\newcommand{\rf}{\mathrm{R}} %right-filtration
\newcommand{\lf}{\mathrm{L}} %left-filtration

\newcommand{\Bb}{\fv{B}}

\newcommand{\col}{\operatorname{Column}}
\newcommand{\row}{\operatorname{Row}}

\newcommand{\rmap}{\longrightarrow}
\newcommand{\lmap}{\longleftarrow}
\newcommand{\lrmap}{\longleftrightarrow}

\newcommand{\fmap}[1]{\stackrel{f_{#1}}{\rmap}}
\newcommand{\gmap}[1]{\stackrel{g_{#1}}{\lmap}}
\newcommand{\pmap}[1]{\stackrel{p_{#1}}{\lrmap}}

\newcommand{\bt}{\mathrm{b}} %birth-time
\newcommand{\dt}{\mathrm{d}} %death-time

\newcommand{\taumod}{{\tau\mathrm{Mod}}}
\newcommand{\filt}{\operatorname{Filt}}

\newcommand{\Hom}{\operatorname{Hom}}
\newcommand{\End}{\operatorname{End}}
\newcommand{\Img}{\operatorname{Im}}
\newcommand{\Ker}{\operatorname{Ker}}
\newcommand{\Coker}{\operatorname{Coker}}
\newcommand{\Span}{\operatorname{Span}}

\newcommand{\Pers}{\operatorname{Pers}}

\begin{document}
%------------------------------------------------------------------
\begin{abstract}
We describe a new methodology for studying persistence of topological features across a family of spaces or point-cloud data sets, called zigzag persistence. Building on classical results about quiver representations, zigzag persistence generalises the highly successful theory of persistent homology and addresses several situations which are not covered by that theory. In this paper we develop theoretical and algorithmic foundations with a view towards applications in topological statistics.
\end{abstract}

%-------------------------------------------------------------------
\maketitle
%-------------------------------------------------------------------
\section{Introduction}
\label{sec:intro}

%--------------------------------------------------
\subsection{Overview}

In this paper, we describe a new methodology for studying persistence of topological features across a family of spaces or point-cloud data sets. This theory of \emph{zigzag persistence} generalises the successful and widely used theory of persistence and persistent homology~\cite{Edelsbrunner_L_Z_2002,Zomorodian_Carlsson_2005}. 
Moreover, zigzag persistence can handle several important situations that are not currently addressed by standard persistence.

The zigzag persistence framework is activated whenever one constructs a \emph{zigzag diagram} of topological spaces or vector spaces: a sequence of spaces $S_1, \dots, S_n$ where each adjacent pair is connected by a map $S_i \to S_{i+1}$ or $S_{i} \leftarrow S_{i+1}$. The novelty of our approach is that the direction of each linking map is arbitrary, in contrast to the usual theory of persistence where all maps point in the same direction.

This paper has three principal objectives:
\begin{itemize}
\item
To describe several scenarios in applied topology where it is natural to consider zigzag diagrams (Section~\ref{sec:intro}).
\item
To develop a mathematical theory of persistence for zigzag diagrams (Sections \ref{sec:zigzags} and~\ref{sec:zigzag2filt}).
\item
To develop algorithms for computing zigzag persistence (Section~\ref{sec:algorithms}).
\end{itemize}
There is one subsidiary objective:
\begin{itemize}
\item
To introduce the \emph{Diamond Principle}, a calculational tool analogous in power and effect to the Mayer--Vietoris theorem in classical algebraic topology (Section~\ref{sec:further}).
\end{itemize}
This is a theoretical paper rather than an experimental paper, and we devote most of our effort to covering the mathematical foundations adequately. The technical basis for zigzag persistence comes from the theory of graph representations, also known as quiver theory. We are deeply indebted to the practitioners of that theory; what is new here is the emphasis on algorithmics and on applications to topology (particularly Sections \ref{sec:intro}, \ref{sec:algorithms} and~\ref{sec:further}).

%--------------------------------------------------
\subsection{Persistence}

One of the principal challenges when attempting to apply algebraic topology to statistical data is the fact that traditional invariants --- such as the Betti numbers or the fundamental group --- are extremely non-robust when it comes to discontinuous changes in the space under consideration. Persistent homology~\cite{Edelsbrunner_L_Z_2002,Zomorodian_Carlsson_2005} is the single most powerful existing tool for addressing this problem.

A typical workflow runs as follows~\cite{deSilva_Carlsson_2004}. The input is a \emph{point cloud}, that is, a finite subset of some Euclidean space or more generally a finite metric space. After an initial filtering step (to remove undesirable points or to focus on high-density regions of the data, say), a set of vertices is selected from the data, and a simplicial complex~$S$ is built on that vertex set, according to some prearranged rule. In practice, the simplicial complex depends on a coarseness parameter~$\epsilon$, and what we have is a nested family $\{ S_\epsilon\}_{\epsilon \in [0, \infty]}$, which typically ranges from a discrete set of vertices at~$S_0$ to a complete simplex at~$S_\infty$.

Persistent homology takes the entire nested family $\{S_\epsilon\}$ and produces a \emph{barcode} or \emph{persistence diagram} as output. A barcode is a collection of half-open subintervals $[b_j,d_j) \subseteq [0,\infty)$, which describes the homology of the family as it varies over~$\epsilon$. An interval $[b_j,d_j)$ represents a homological feature which is born at time~$b_j$ and dies at time~$d_j$. This construction has several excellent properties:
\begin{itemize}
\item
There is no need to select a particular value of~$\epsilon$.
\item
Features can be evaluated by interval length. Long intervals are expected to indicate essential features of the data, whereas short intervals are likely to be artefacts of noise.
\item
There exists a fast algorithm to compute the barcode~\cite{Zomorodian_Carlsson_2005}.
\item
The barcode is a complete invariant of the homology of the family of complexes~\cite{Zomorodian_Carlsson_2005}.
\item
The barcode is provably stable with respect to changes in the input~\cite{CohenSteiner_E_H_2007}. In contrast, any individual homology group $H_k(S_\epsilon)$ is highly unstable.
\end{itemize}

The major limitation of persistence is that it depends crucially on the family $\{S_\epsilon\}$ being nested, in the sense that $S_\epsilon \subseteq S_{\epsilon'}$ whenever $\epsilon \leq \epsilon'$.
This applies to the current theoretical understanding as well as the algorithms.
Zigzag persistence addresses this limitation.

If we discretise the variable~$\epsilon$ to a finite set of values, the family of simplicial complexes can be thought of as a diagram of spaces
\[
S_1 \to S_2 \to \dots \to S_n
\]
where the arrows represent the inclusion maps.
If we apply the $k$-dimensional homology functor $H_k(_; \Ff)$ with coefficients in a field~$\Ff$, this becomes a diagram of vector spaces
\[
V_1 \to V_2 \to \dots \to V_n
\]
and linear maps, where $V_i = H_k(S_i; \Ff)$. Such a diagram is called a \emph{persistence module}. What makes persistence work is that there is a simple algebraic classification of persistence modules up to isomorphism; each possible barcode corresponds to an isomorphism type.

Our goal is to achieve a similar classification for diagrams in which the arrows may point in either direction. This is zigzag persistence, in a nutshell.

%--------------------------------------------------
\subsection{Zigzag diagrams in applied topology}

We consider some problems which arise quite naturally in the computational topology of data.  

\begin{Example}
\label{densityprofile}
Some of the most interesting properties of a point cloud are contained in the estimates of the probability density from which the data are sampled. Deep structure is sometimes revealed after thresholding according to a density estimate (see \cite{Carlsson_I_dS_Z_2008} for an example drawn from visual image analysis). However, the construction of a density estimation function~$\rho$ invariably depends on choosing a smoothing parameter: for instance $\rho(x)$ might be defined to be the number of data points within distance~$r$ of~$x$; here $r$ is the smoothing parameter.

It happens that different choices of smoothing parameter may well reveal different structures in the data; a particularly striking example of this occurs in~\cite{Carlsson_I_dS_Z_2008}.
Statisticians have invented useful criteria for determining what the `appropriate' value of such a parameter might be for a particular data set;  but another point of view would be to analyse all values of the parameter simultaneously, and to study how the topology changes as the parameter varies.

The problem with doing this is that there is no natural relationship between, say, the $25\% $ densest  points as measured using two different parameter values.  This means that one cannot build an increasing family of spaces using the change in parameters, and so one cannot use persistence to analyze the evolution of the topology.  On the other hand, there are natural zigzag sequences which can be used to study this problem. Select a sequence of parameter values $r_1 < r_2 <  \dots<  r_n$ and a percentage~$p$, and let $X_r^p$ denote the densest $p\%$ of the point cloud when measured according to parameter value~$r$. We can then consider the union sequence
\[
X^p_{r_1}
\rightarrow
X^p_{r_1} \cup X^p_{r_2}
\leftarrow
X^p_{r_2}
\rightarrow
X^p_{r_2} \cup X^p_{r_3}
\leftarrow
X^p_{r_3}
\rightarrow
\dots
\leftarrow
X^p_{r_n}
\]
or the intersection sequence
\[
X^p_{r_1}
\leftarrow
X^p_{r_1} \cap X^p_{r_2}
\rightarrow
X^p_{r_2}
\leftarrow
X^p_{r_2} \cap X^p_{r_3}
\rightarrow
X^p_{r_3}
\leftarrow
\dots
\rightarrow
X^p_{r_n}.
\]
As we see in Section~\ref{subsec:strongdiamond}, there is essentially no difference between the zigzag persistent homology of the union and intersection sequences of a sequence of spaces. Here that assertion needs to be filtered through the process of representing the data subsets~$X_r^p$ as simplicial complexes.
\end{Example}

\begin{Example}[Topological bootstrapping]
\label{sampling}
Suppose we are given a very large point cloud~$X$. If it is too large to process directly, we may take a sequence of small samples $X_1, \dots, X_n$ and estimate their topology individually, perhaps obtaining a persistence barcode for each one. How does this reflect the topology of the original sample~$X$? On one hand, if most of the barcodes have similar appearance, then one might suppose that $X$ itself will have the same barcode. On the other hand, one needs to be able to distinguish between a single feature detected repeatedly, and multiple features detected randomly but one at a time. If we detect $n$~features in~$X_i$ on average, are we detecting $n$~features of~$X$ with detection probability~1, or $kn$ features with detection probability~$1/k$? 

Once again, there is a need to correlate features across different instances of the construction. The union sequence comes to the rescue:
\[
X_1
\rightarrow
X_1 \cup X_2 
\leftarrow
X_2
\rightarrow
\dots
\leftarrow
X_n  
\]
In this case, the intersection sequence is not useful at the level of samples, because two sparse samples are unlikely to intersect
very much.

The approach in this example is analogous to bootstrapping in statistics, where measurements on a large data set are estimated by making repeated measurements on a set of samples.
\end{Example}

\begin{Example}
\label{landmarking}
In computational topology, there exist several techniques for modelling a point cloud data set~$X$ by a simplicial complex~$S$: the Cech complex, the Vietoris--Rips complex, the alpha complex~\cite{Edelsbrunner_Mucke_1994}, the witness complex~\cite{deSilva_Carlsson_2004}, and so on. The witness complex $W(X; L)$, in particular, depends on the choice of a small subset of `landmark' points $L \subset X$ which will serve as the vertex set of~$S$.
Roughly speaking, a simplex~$\sigma$ with vertices in~$L$ is included in $W(X;L)$ if there is some $x \in X$ which witnesses it, by being close to all the vertices.

How does the witness complex $W(X; L)$ depend on the choice of landmark set? There is no direct way to compare $W(X; L)$ with $W(X; M)$ for two different choices of landmark sets $L, M$. However, it turns out that one can define a witness \emph{bicomplex} $W(X; L, M)$ which maps onto each witness complex. The cells are cartesian products $\sigma \times \tau$, where $\sigma, \tau$ have vertices in $L, M$ respectively. A cell $\sigma \times \tau$ is included provided that there exists $x \in X$ which simultaneously witnesses $\sigma$ for $W(X;L)$ and $\tau$ for $W(X;M)$.

Given a sequence $L_1, \dots, L_n$ of landmark subsets, one can then construct the biwitness zigzag:
\[
W(X; L_1)
\leftarrow
W(X; L_1, L_2)
\rightarrow
W(X; L_2)
\leftarrow
\dots
\rightarrow
W(X; L_n)
\]
Long intervals in the zigzag barcode will then indicate features that persist across the corresponding range of choices of landmark set.
\end{Example}

The fundamental requirement is then for a way of assessing, in a zigzag diagram of vector spaces, the degree to which consistent families of elements exist.  The point  of this paper is that there is such methodology.  We will interpret the isomorphism classes of zig-zag diagrams as a special case of the classification problem for quiver representations (see \cite{Derksen_Weyman_2005} for background on this theory). There turns out to be a  theorem of Gabriel~\cite{Gabriel_1972} which classifies arbitrary diagrams based on Dynkin diagrams, and which shows in particular that the set of isomorphism classes of zigzag diagrams of a given length is parametrised by barcodes --- just as persistence modules are. Long intervals in the classification define large families of consistent elements, hence indicate the presence of features stable across samples, landmark sets,  or parameter values for a density estimator.

%--------------------------------------------------
\subsection{Organisation of the paper}

In Section~\ref{sec:zigzags} we describe the theory of decompositions of zigzag modules. These decompositions produce zigzag persistence barcodes analogous to the barcodes of persistent homology. The foundational theorem of Gabriel is stated without proof.
In Section~\ref{sec:zigzag2filt} we develop the machinery of right-filtrations, which turn out to be the right tool for accessing the decomposition structure of a zigzag module. This is an important section for the reader who wishes to make serious use of zigzag persistence.
In Section~\ref{sec:algorithms}, we present a general-purpose algorithmic framework for calculating zigzag persistence, and we show how this operates in a practical class of examples. The algorithm is based on a proof of Gabriel's theorem for zigzag modules, included for completeness.
Section~\ref{sec:further} is devoted to a localisation principle which gives another approach to zigzag barcode calculations. We apply this to prove the Diamond Principle. We use this in turn to compare the zigzag barcodes for two natural zigzag diagrams obtained from a sequence of simplicial complexes.
 
%------------------------------------------------------------------
\section{Zigzag Diagrams of Vector Spaces}
\label{sec:zigzags}

We work over a field $\Ff$ which remains fixed throughout this paper. There is no significance to the choice of~$\Ff$. All vector spaces are finite-dimensional.

%--------------------------------------------------
\subsection{Zigzag modules.}
\label{subsec:taumod}

Let $\Vv$ denote a sequence of vector spaces and linear maps, of length~$n$:
\[
V_1 \stackrel{p_1}{\lrmap}
V_2 \stackrel{p_2}{\lrmap}
\dots
\stackrel{p_{n-1}}{\lrmap}
V_n
\]
Each $\stackrel{p_i}{\lrmap}$ represents either a forward map $\fmap{i}$ or a backward map $\gmap{i}$.
The object $\Vv$ is called a \textbf{zigzag diagram} of vector spaces, or simply a \textbf{zigzag module}, over~$\Ff$.

The sequence of symbols $f$ or $g$ is the \textbf{type} of~$\Vv$. For instance, a diagram of type $\tau = fgg$ looks like this:
\[
V_1 \fmap{1}
V_2 \gmap{2}
V_3 \gmap{3}
V_4
\]
The length of a type~$\tau$ is the length of any diagram of type~$\tau$. For example, we say that $fgg$ has length~4.
We will usually be considering zigzag modules of a fixed type~$\tau$ of length~$n$. Such diagrams are called \textbf{$\tau$-modules}, and the class of $\tau$-modules is denoted $\taumod$.

\textbf{Persistence modules} (see~\cite{Edelsbrunner_L_Z_2002,Zomorodian_Carlsson_2005}) are zigzag modules where all the maps have the forward orientation; in other words, where $\tau = ff\dots f$.
As explained in~\cite{Zomorodian_Carlsson_2005}, persistence modules can be viewed as graded modules over the polynomial ring $\Ff[t]$.  This observation simplifies the analysis of persistence modules quite considerably.

More generally, one can consider \textbf{graph representations} of arbitrary oriented graphs. Zigzag modules constitute the special case where the graph is $A_n$ (a path with $n$~vertices and $n-1$ edges) and the orientation is specified by the type~$\tau$.
In 1972, Gabriel showed that the Dynkin--Coxeter graphs $A_n$, $D_n$, $E_6$, $E_7$, $E_8$ (arbitrarily oriented) have an especially well-behaved representation theory~\cite{Gabriel_1972}. 
The theory of \textbf{quivers} was launched from this starting block; see~\cite{Derksen_Weyman_2005} for a beautiful and transparent introduction.
Zigzag persistence is enabled by the good behaviour of $A_n$ graph representations.

\begin{Remark}
$\taumod$ has the structure of an abelian category. Given two $\tau$-modules $\Vv, \Ww$, a morphism $\alpha: \Vv \to \Ww$ is defined to be a collection of linear maps $\alpha_i: V_i \to W_i$ which satisfy the commutation relations $\alpha_{i+1} f_i = h_i \alpha_i$ or $\alpha_i g_i = k_i \alpha_{i+1}$ for each~$i$. (Here the forward and backward maps for~$\Ww$ are written $h,k$ respectively.)
Morphisms can be composed in the obvious way, and have kernels, images, and cokernels: for instance $\tm{K} = \Ker(\alpha)$ is the $\tau$-module with spaces $K_i = \Ker(V_i \to W_i)$ and maps $f_i |_{K_i}$ and $g_i |_{K_{i+1}}$ defined by restriction. 
The set of morphisms $\Hom(\Vv,\Ww)$ is naturally a vector space over~$\Ff$, and the endomorphism ring $\End(\Vv) = \Hom(\Vv,\Vv)$ is a non-commutative $\Ff$-algebra. We can view $\End(\Vv)$ as the subalgebra of $\End(V_1) \times \dots \times \End(V_n)$ defined by the commutation relations.
\end{Remark}

%--------------------------------------------------
\subsection{Decompositions of zigzag modules.}
\label{subsec:taudecomp}

We wish to understand zigzag modules by decomposing them into simpler parts. Accordingly, a \textbf{submodule} $\Ww$ of a $\tau$-module~$\Vv$ is defined by subspaces $W_i \leq V_i$ such that $f_i(W_i) \subseteq W_{i+1}$ or $g_i(W_{i+1}) \subseteq W_i$ for all~$i$. These conditions guarantee that $\Ww$ is itself a $\tau$-module, with maps given by the restrictions $f_i |_{W_i}$ or $g_i | _{W_{i+1}}$. We write $\Ww \leq \Vv$.

%More strongly, 
A submodule~$\Ww$ is called a \textbf{summand} of~$\Vv$ if there exists a submodule $\Xx \leq \Vv$ which is complementary to~$\Ww$, in the sense that $V_i = W_i \oplus X_i$ for all~$i$.
%
%As a rule, $\tau$-modules have many more submodules than summands.
%
In that case, we say that $\Vv$ is the \textbf{direct sum} of $\Ww, \Xx$ and write $\Vv = \Ww \oplus \Xx$.

\begin{Example}
As a rule, most submodules are not summands. $\Vv = (\Ff \stackrel{1}{\rmap} \Ff)$ has the submodule $\Ww = (0 \rmap \Ff)$. However, $\Ww$ is not a summand because the only possible complement is $(\Ff \rmap 0)$, and that is not a submodule of~$\Vv$.
\end{Example}

\begin{Remark}
The direct sum can also be defined as an `external' operation: given $\tau$-modules $\Vv, \Ww$ their direct sum $\Vv \oplus \Ww$ is defined to be the $\tau$-module with spaces $V_i \oplus W_i$ and maps $f_i \oplus h_i$ or $g_i \oplus k_i$.
(Here the forward and backward maps for~$\Ww$ are written $h,k$ respectively.)
\end{Remark}

A $\tau$-module $\Vv$ is \textbf{decomposable} if it can be written as a direct sum of nonzero submodules, and \textbf{indecomposable} otherwise. 
Any $\tau$-module $\Vv$ has a \textbf{Remak decomposition}; in other words we can write $\Vv = \Ww_1 \oplus \dots \oplus \Ww_N$, where the summands $\Ww_j$ are indecomposable. The existence of such a decomposition is proved by induction on the total dimension $\sum_i \dim(V_i)$: if $\Vv$ is decomposable, say $\Vv = \Ww \oplus \Xx$, then we may assume inductively that $\Ww, \Xx$ have Remak decompositions, and therefore so does~$\Vv$.
(Base case: if $\Vv$ is indecomposable, then it has a Remak decomposition with one term.)

Remak decompositions themselves are not unique. However, the Krull--Schmidt principle from commutative algebra tells us that the summands in a Remak decomposition are unique up to reordering:

\begin{Proposition}(Krull--Remak--Schmidt.)
\label{prop:krull}
Suppose a $\tau$-module $\Vv$ has Remak decompositions
\[
\Vv = \Ww_1 \oplus \dots \oplus \Ww_M
\quad
\mbox{and}
\quad
\Vv = \Xx_1 \oplus \dots \oplus \Xx_N.
\]
Then $M = N$ and there is some permutation~$\sigma$ of $\{1, \dots, N\}$ such that $\Ww_j \cong \Xx_{\sigma(j)}$ for all~$j$.
\end{Proposition}

\begin{proof}
The proof of Theorem~7.5 of Lang~\cite{Lang_2005}, which is stated for modules in the ordinary sense, can be applied verbatim to our present context; all the required algebraic operations can be carried out within $\End(\Vv)$. Since our $\tau$-modules have finite total dimension, the ascending and descending chain conditions (\textsc{acc} and \textsc{dcc}) are automatic.
\end{proof}

For further context, we refer the reader to an elegant article by Atiyah~\cite{Atiyah_1956}; the Krull--Schmidt principle applies in any exact abelian category to objects which satisfy \textsc{acc} and \textsc{dcc}, or a weaker `bi-chain condition' defined by Atiyah. Our category, $\taumod$, is included by this formulation. 

Thus we can use the multiset $\{ \Ww_j \}$ as an isomorphism invariant of~$\Vv$.
For this to be useful, we need to identify the set of indecomposable $\tau$-modules.
We now describe a natural collection of indecomposables. For each subinterval $[b,d]$ of the integer sequence $\{1, \dots, n\}$ there is an associated $\tau$-module.

\begin{Definition}
Let $\tau$ be a type of length~$n$ and consider integers $1 \leq b \leq d \leq n$. The \textbf{interval} $\tau$-module with birth time~$b$ and death time~$d$ is written $\Ii_\tau(b,d)$ and defined with spaces
\[
I_i = \left\{ 
\begin{array}{ll}
\Ff \quad & \mbox{if $b \leq i \leq d$,}
\\
0 \quad & \mbox{otherwise;}
\end{array}
\right.
\]
and with identity maps between adjacent copies of $\Ff$, and zero maps otherwise. When $\tau$ is implicit, we will usually suppress it and simply write $\Ii(b,d)$.
\end{Definition}

\begin{Example*}
If $\tau = fgg$ then $\Ii(2,3)$ is the zigzag module
\[
0 \stackrel{0}{\rmap}
\Ff \stackrel{1}{\lmap}
\Ff \stackrel{0}{\lmap}
0.
\]
\end{Example*}

\begin{Proposition}
\label{prop:intind}
Interval $\tau$-modules are indecomposable.
\end{Proposition}

\begin{proof}
Suppose $\Ii(b,d) = \Vv \oplus \Ww$ and consider two adjacent terms~$\Ff$ connected by an identity map. Since $\Vv, \Ww$ are submodules, the dimensions of $\Vv$ and $\Ww$ cannot decrease in the direction of the map; nor, since they are complements, can they increase. Thus $\dim(V_i)$ and $\dim(W_i)$ are constant over $b \leq i \leq d$, and in particular one of $\Vv,\Ww$ must be zero.
\end{proof}

Here is the foundation stone for the theory of zigzag persistence.

\begin{Theorem}[Gabriel]
\label{thm:interval}
The indecomposable $\tau$-modules are precisely the intervals $\Ii(b,d)$, where $1 \leq b \leq d \leq n = \operatorname{length}(\tau)$. Equivalently, every $\tau$-module can be written as a direct sum of intervals.
\end{Theorem}

\begin{proof}
This is the simplest special case of Gabriel's theorem, for the graphs $A_n$. The original reference (in German) is~\cite{Gabriel_1972}. See~\cite{Derksen_Weyman_2005} for an accessible overview.
\end{proof}

Thus, any $\tau$-module can be described completely up to isomorphism as an unordered list of intervals $[b,d]$, which correspond to its indecomposable summands.
This is in exact accordance with the special case of ordinary persistence, where the result is comparatively easy to prove: it is simply the classification of finitely-generated graded modules over the polynomial ring $\Ff[t]$ (see~\cite{Zomorodian_Carlsson_2005}).

The philosophical point is that the decomposition theory of graph representations is somewhat independent of the orientation of the graph edges (see Kac~\cite{Kac_1980}). Even in our case this is surprising, because there is no obvious  congruence between persistence modules and zigzag modules of an arbitrary type~$\tau$. However, if we accept this principle, then the generalisation from ordinary persistence to zigzag persistence is not surprising: interval decomposition for persistence modules implies interval decomposition for zigzag modules.

We will devote much of this paper to constructing a stand-alone proof of Theorem~\ref{thm:interval}. This provides technical support towards our two main goals: to provide algorithms for computing the interval summands of a given $\tau$-module; and to make rigorous statements about the output of those algorithms.

%--------------------------------------------------
\subsection{Zigzag persistence.}
\label{subsec:subsum}

We now define zigzag persistence and develop some of its elementary properties.

\begin{Definition}
Let $\Vv$ be a zigzag module (of arbitrary type). The \textbf{zigzag persistence} of~$\Vv$ is defined to be the multiset
\[
\Pers(\Vv) =
\left\{ [b_j, d_j] \subseteq \{1, \dots, n\} \mid  j = 1, \dots, N \right\}
\]
of integer intervals derived from a decomposition $\Vv \cong \Ii(b_1,d_1) \oplus \cdots \oplus \Ii(b_N, d_N)$. The Krull--Schmidt principle asserts that this definition is independent of the decomposition.
\end{Definition}

Graphically, $\Pers(\Vv)$ can be represented as a set of lines measured against a single axis with labels $\{1, \dots, n\}$ (the \textbf{barcode}), or as a multiset of points in~$\mathbb{R}^2$ lying on or above the diagonal in the positive quadrant (the \textbf{persistence diagram}). See Figure~\ref{fig:barpd} for an example presented in each style.

\begin{figure}
\centerline{
\includegraphics[scale=0.5]{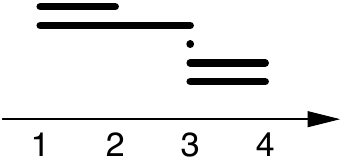}
\qquad\qquad
\includegraphics[scale=0.5]{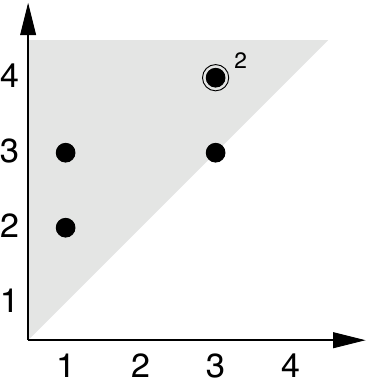}
}
\caption{Barcode (left) and persistence diagram (right) representations of the persistence $\{ [1,2], [1,3], [3,3], [3,4], [3,4] \}$ of a zigzag module of length~4.}
\label{fig:barpd}
\end{figure}

\begin{Remark}
In the special case of persistence modules, this agrees with the standard treatment (see~\cite{Edelsbrunner_L_Z_2002, Zomorodian_Carlsson_2005}) {except} in the following particular: the closed integer intervals $[b_j,d_j] \subseteq \{1, \dots, n\}$ are replaced by half-open real intervals $[b_j,d_j+1) \subset \mathbb{R}$ in the standard treatment. This is particularly natural when the indexing parameter is continuous: an interval $[b,d)$ indicates a feature born at time~$b$ that survives right up to, but vanishes at, time~$d$.
Our convention is motivated by the desire to maintain symmetry between the forward and backward directions.
We advise the reader to take particular care in handling the different conventions.
\end{Remark}

The transition from a zigzag module to its interval decomposition presents certain hazards which are not present in the case of persistence modules. We now draw attention to these hazards.

\begin{Definition}
Let $\Vv$ be a zigzag module and let $\Vv[p,q]$ denote the restriction of~$\Vv$ to the index set $p \leq i \leq q$. A \textbf{feature} of~$\Vv$ over the time interval $[p,q]$ is a summand of $\Vv[p,q]$ isomorphic to $\Ii(p,q)$. 
\end{Definition}

With persistence modules, there are several equivalent ways to recognise the existence of a feature. Here is a sample result.

\begin{Proposition}
\label{prop:intuitions}
Let $\Vv$ be a persistence module of length~$n$, and let $1 \leq p \leq q \leq n$. The following are equivalent:
\begin{enumerate}
\item
The composite map $V_p \to V_q$ is nonzero.
\item
There exist nonzero elements $x_i \in V_i$ for $p \leq i \leq q$, such that $x_{i+1} = f_i(x_i)$ for $p \leq i < q$.
\item
There exists a submodule of $\Vv[p,q]$ isomorphic to $\Ii(p,q)$.
\item
There exists a summand of $\Vv[p,q]$ isomorphic to $\Ii(p,q)$, i.e.\ a feature over $[p,q]$.
\end{enumerate}
\end{Proposition}

\begin{proof}
It is easy to verify that (1), (2), (3) are equivalent. For (1)~$\Rightarrow$~(2), begin by choosing $x_p \in V_p$ that maps to a nonzero element in $V_q$, and let $x_i$ be the image of $x_p$ in $V_i$. For (2)~$\Rightarrow$~(3), define $\Ii$ by $I_i = \Span(x_i)$. For (3)~$\Rightarrow$~(1), note that the restriction $I_p \to I_q$ is nonzero.

Clearly (4)~$\Rightarrow$~(3). We now show that (1)~$\Rightarrow$~(4). Consider an interval decomposition $\Vv[p,q] = \Ii(b_1, d_1) \oplus \dots \oplus \Ii(b_N,d_N)$. On each summand, the map $I_p(b_j,d_j) \to I_q(b_j,d_j)$ is zero unless $b_j = p$ and $d_j = q$. Thus at least one of the summands is isomorphic to $\Ii(p,q)$.
\end{proof}

The intuitions supported by Proposition~\ref{prop:intuitions} break down in the general case.

\begin{Caution}
Let $\Vv$ be a zigzag module of arbitrary type. Statement~(1) has no clear interpretation at this stage (something can be said in terms of the right-filtration functor of Section~\ref{sec:zigzag2filt}). Consider the following statements:
\begin{enumerate}
\setcounter{enumi}{1}
\item
There exist nonzero elements $x_i \in V_i$ for $p \leq i \leq q$, such that $x_{i+1} = f_i(x_i)$ or $x_i = g_i(x_{i+1})$ (whichever is applicable) for $p \leq i < q$.
\item
There exists a submodule of $\Vv[p,q]$ isomorphic to $\Ii(p,q)$.
\item
There exists a summand of $\Vv[p,q]$ isomorphic to $\Ii(p,q)$, i.e.\ a feature over $[p,q]$.
\end{enumerate}
It is easy to verify that $(2) \Leftrightarrow (3)$ and that (4) implies $(2), (3)$. However, the next two examples demonstrate that $(2), (3)$ do not in general imply~(4).
\end{Caution}

\begin{Example}
\label{ex:caution1}
Let $\tau = gf$ and consider the $\tau$-module $\Vv$ defined as follows:
\[
\begin{diagram}
\node{\Ff}
\node{\Ff^2}
\arrow{w,t}{g_1}
\arrow{e,t}{f_2}
\node{\Ff}
\\
\node{x}
\node{(x,y)}
\arrow{w,T}
\arrow{e,T}
\node{y}
\end{diagram}
\]
The interval decomposition is $\Vv = \Ii(1,2) \oplus \Ii(2,3)$, where the summands are
\[
\begin{diagram}
\node{\Ff}
\node{\Ff \oplus 0}
\arrow{w,t}{}
\arrow{e,t}{}
\node{0}
\\
\node{x}
\node{(x,0)}
\arrow{w,T}
%\arrow{e,T}
\node{}
\end{diagram}
\]
and
\[
\begin{diagram}
\node{0}
\node{0 \oplus \Ff}
\arrow{w,t}{}
\arrow{e,t}{}
\node{\Ff}
\\
\node{}
\node{(0,y)}
%\arrow{w,T}
\arrow{e,T}
\node{y}
\end{diagram}
\]
respectively. If this example appeared in a statistical topology setting,  the feature corresponding to the generator of the $\Ff$ at~$V_1$ would be regarded as unrelated to the feature corresponding to the generator of the $\Ff$ at~$V_3$.

On the other hand, $\Vv$ does have a submodule (in fact, many submodules) isomorphic to $\Ii(1,3)$. Indeed, let $\Delta = \{ (x,x) \mid x \in \Ff\}$ denote the diagonal subspace of $\Ff^2$.
Then
\[
\begin{diagram}
\node{\Ff}
\node{\Delta}
\arrow{w}
\arrow{e}
\node{\Ff}
\\
\node{x}
\node{(x,x)}
\arrow{w,T}
\arrow{e,T}
\node{x}
\end{diagram}
\]
is a submodule $\Ww \leq \Vv$ isomorphic to $\Ii(1,3)$. The quotient $\tau$-module $\Vv / \Ww$ is isomorphic to $\Ii(2,2)$ but $\Ww$ has no complementary $\tau$-module in~$\Vv$. Indeed, that would contradict the Krull--Schmidt theorem. More concretely, any complement of~$\Ww$ must be isomorphic to $(0 \lmap \Ff \rmap 0)$, but that would require a 1-dimensional subspace of $\Ker(g_1) \cap \Ker(f_2) = 0$.
\end{Example}

\begin{Example}
\label{ex:caution2}
We can extend the previous example to arbitrary length. Consider the type $\tau = gf\dots gf = 
(gf)^n$, of length $2n+1$. Let $\Vv$ be the $\tau$-module
\[
\Ff \stackrel{\pi_1}{\lmap}
\Ff^2 \stackrel{\pi_2}{\rmap}
\Ff \stackrel{\pi_1}{\lmap}
\cdots
\stackrel{\pi_2}{\rmap}
\Ff \stackrel{\pi_1}{\lmap}
\Ff^2 \stackrel{\pi_2}{\rmap}
\Ff,
\]
where $\pi_1(x,y) = x$, and $\pi_2(x,y) = y$.
Then $\Vv$ is isomorphic to a sum of short intervals
\[
\Ii(1,2) \oplus
\left\{ \Ii(2,4) \oplus \dots \oplus \Ii(2n-2,2n) \right\}
\oplus \Ii(2n,2n+1)
\]
but it has a submodule
\[
\Ff \stackrel{}{\lmap}
\Delta \stackrel{}{\rmap}
\Ff \stackrel{}{\lmap}
\cdots
\stackrel{}{\rmap}
\Ff \stackrel{}{\lmap}
\Delta \stackrel{}{\rmap}
\Ff
\]
isomorphic to the long interval $\Ii(1,2n+1)$.
\end{Example}

\begin{Moral}
In zigzag persistence it is necessary to respect the distinction between submodules and summands. Features are defined in terms of summands; never submodules.
\end{Moral}

We have defined features in terms of a chosen subinterval $[p,q]$. Features behave as expected when zooming to a larger or smaller window of observation. The following proposition illustrates what we mean.

\begin{Proposition}
\label{prop:restriction1}
Let $\Vv$ be a zigzag module of length~$n$ and let $1 \leq p \leq q \leq n$. The following statements are equivalent.
\begin{enumerate}
\item
There exists a summand of $\Vv[p,q]$ isomorphic to $\Ii(p,q)$, i.e.\ a feature over $[p,q]$.
\item
There exists a summand of $\Vv$ isomorphic to $\Ii(p',q')$, for some 
$[p',q'] \supseteq [p,q]$.
%$p' \leq p$, $q \leq q'$.
\end{enumerate}
Indeed, there is a bijection between intervals $[p,q]$ in~$\Pers(\Vv[p,q])$ and intervals $[p',q'] \supseteq [p,q]$ in $\Pers(\Vv)$.
\end{Proposition}

\begin{proof}
Consider an interval decomposition $\Vv = \Ii(b_1, d_1) \oplus \cdots \oplus \Ii(b_N, d_N)$. By restriction, this induces an interval decomposition of $\Vv[p,q]$ into intervals $\Ii(b_j,d_j)[p,q]$. This induces the claimed bijection, because $[b_j,d_j]$ restricts to $[p,q]$ if and only if $[b_j,d_j] \supseteq [p,q]$.
\end{proof}

Operating invisibly in this proof is the Krull--Schmidt principle, which allows us to select the interval decompositions most convenient to us when calculating $\Pers(\Vv)$ and $\Pers(\Vv[p,q])$.

\begin{Remark}
Sometimes it is useful to reduce the resolution of $\Pers(V)$. Let $K \subset \{1, \dots, n\}$ be any subset. We define the \textbf{restriction} of $\Pers(\Vv)$ to~$K$ to be the multiset
\[
\Pers(\Vv)|_K =
\left\{
I \cap K
\mid
I \in \Pers(\Vv),\, I \cap K \not= \emptyset
\right\}.
\]
For instance, Proposition~\ref{prop:restriction1} amounts to the observation that $\Pers(\Vv[p,q]) = \Pers(\Vv)|_{[p,q]}$.
\end{Remark}

%------------------------------------------------------------------
\section{From Zigzag Modules to Filtrations}
\label{sec:zigzag2filt}

%--------------------------------------------------
\subsection{The right-filtration operator}
\label{subsec:rfilt}

Our strategy is to understand (and construct) decompositions of a $\tau$-module~$\Vv$ by an iterative process, moving from left to right and retaining the necessary information at each stage. The bulk of this information is encoded as a filtration on the rightmost vector space $V_n$.

\begin{Definition}
\label{def:rfilt}
The {\bf right-filtration} $\rf(\Vv)$ of a $\tau$-module $\Vv$ of length~$n$ takes the form
\[
\rf(\Vv) = (R_0, R_1, \dots, R_n),
\]
where the~$R_i$ are subspaces of~$V_n$ satisfying the inclusion relations
\[
0 = R_0 \leq R_1 \leq \dots \leq R_n = V_n.
\]
$\rf(\Vv)$ is defined recursively as follows.

Base case:
\begin{itemize}
\item
If $\Vv$ has length~1, then $\rf(\Vv) = (0, V_1)$.
\end{itemize}

Recursive step. Suppose we have already defined $\rf(\Vv) = (R_0, R_1, \dots, R_n)$:
\begin{itemize}
\item
If $\Vv^+$ is $\Vv \fmap{n} V_{n+1}$, then $\rf(\Vv^+) = (f_n(R_0), f_n(R_1), \dots, f_n(R_n), V_{n+1})$.
\item
If $\Vv^+$ is $\Vv \gmap{n} V_{n+1}$, then $\rf(\Vv^+) = (0, g_n^{-1}(R_0), g_n^{-1}(R_1), \dots, g_n^{-1}(R_n))$.
\end{itemize}
To verify that $\rf(\Vv^+)$ in the two inductive cases is a filtration of the specified form, note that $R_i \leq R_{i+1}$ implies that $f_n(R_i) \leq f_n(r_{i+1})$ in the first case, and $g_n^{-1}(R_i) \leq g_n^{-1}(R_{i+1})$ in the second case. Moreover $f_n(R_0) = f_n(0) = 0$, and $g_n^{-1}(R_n) = g_n^{-1}(V_n) = V_{n+1}$.
\end{Definition}

\begin{Example} Here are the right-filtrations for the two length-2 types:
%Here are the two length-2 cases:
\begin{eqnarray*}
\rf(V_1 \fmap{1} V_2)
 &=& (0, f_1(V_1), V_2)
\\
\rf(V_1 \gmap{1} V_2)
 &=& (0, g_1^{-1}(0), V_2)
\end{eqnarray*}
\end{Example}

\begin{Example} Here are the right-filtrations for the four length-3 types:
%Here are the four length-3 cases:
\begin{eqnarray*}
\rf(V_1 \fmap{1} V_2 \fmap{2} V_3)
 &=& (0, f_2 f_1(V_1), f_2(V_2), V_3)
\\
\rf(V_1 \fmap{1} V_2 \gmap{2} V_3)
 &=& (0, g_2^{-1}(0), g_2^{-1} f_1(V_1), V_3)
\\
\rf(V_1 \gmap{1} V_2 \fmap{2} V_3)
 &=& (0, f_2 g_1^{-1}(0), f_2(V_2), V_3)
\\
\rf(V_1 \gmap{1} V_2 \gmap{2} V_3)
 &=& (0, g_2^{-1}(0), g_2^{-1} g_1^{-1} (0), V_3)
\end{eqnarray*}
See Figure~\ref{fig:rfilt3} for a schematic representation.
\end{Example}

\begin{figure}
\centerline{
\includegraphics[scale=0.5]{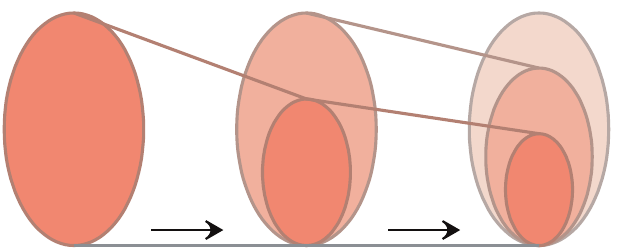}
\qquad
\includegraphics[scale=0.5]{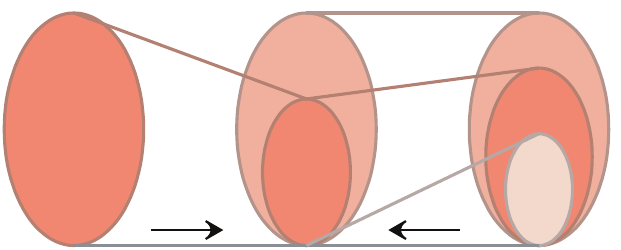}
}
\vspace{3ex}
\centerline{
\includegraphics[scale=0.5]{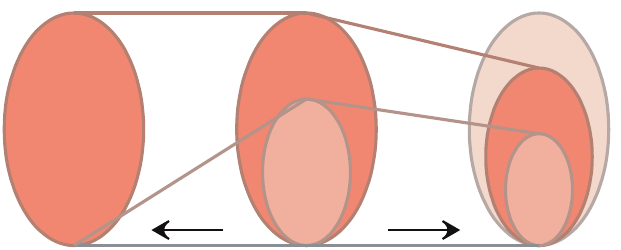}
\qquad
\includegraphics[scale=0.5]{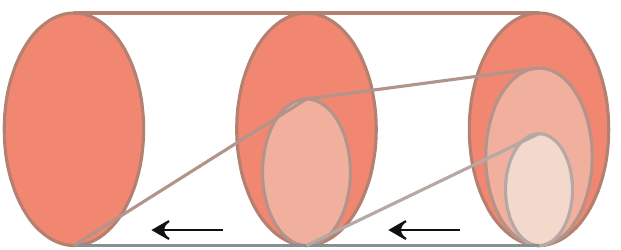}
}
\caption{Forward propagation of the right-filtration, illustrated for the four types $ff$, $fg$, $gf$, $gg$ of length~3.}
\label{fig:rfilt3}
\end{figure}

\begin{Remark}
In the examples above, it is not difficult to see that $\rf(\Vv)$ comprises all the subspaces of~$V_n$ that are naturally definable in terms of the maps $p_i$.
\end{Remark}

Each of the $n$~subquotients $R_{i}/R_{i-1}$ carries information dating back to some earliest $V_j$ in the sequence of vector spaces.

\begin{Example}
\label{ex:birthtimef}
The module $V_1 \fmap{1} V_2$ has right-filtration $(0, f_1(V_1), V_2)$. The first subquotient $f_1(V_1)/0 = f_1(V_1)$ corresponds to vectors born at time~1 which survive to time~2. The second subquotient $V_2/f_1(V_1)$ corresponds to vectors which appear only at time~2. 
\end{Example}

\begin{Example}
\label{ex:birthtimeg}
The module $V_1 \gmap{1} V_2$ has right-filtration $(0, g_1^{-1}(0), V_2)$. The first subquotient $g_1^{-1}(0)$ corresponds to vectors at time~2 which are destroyed when mapping back to time~1. The second subquotient $V_2 / g_1^{-1}(0)$ is isomorphic to $g_1(V_2)$ and records those vectors which survive from time~2 back to time~1.
\end{Example}

\begin{Definition}
\label{def:birthtime}
The {\bf birth-time index} $\bt(\tau) = (b_1, b_2, \dots, b_n)$ is a vector of integers $b_i$ which indicate the birth-times associated with the subquotients $R_i / R_{i-1}$ of the right-filtration of a $\tau$-module. This is defined recursively as follows.

Base case:
\begin{itemize}
\item
If $\tau$ is empty (i.e. $\Vv$ has length~1) then $\bt(\tau) = (1)$.
\end{itemize}

Recursive step. Suppose we have already defined $\bt(\tau) = (b_1, b_2, \dots, b_n)$:
\begin{itemize}
\item
If $\tau^+$ is $\tau f$ then $\bt(\tau^+) = (b_1, \dots, b_n, n+1)$.
\item
If $\tau^+$ is $\tau g$ then $\bt(\tau^+) = (n+1, b_1, \dots, b_n)$.
\end{itemize}
\end{Definition}

\begin{Example}
At length~2 we have $\bt(f) = (1,2)$ whereas $\bt(g) = (2,1)$. This is consonant with the discussion in Examples $\ref{ex:birthtimef}$ and~$\ref{ex:birthtimeg}$.
\end{Example}

\begin{Example}
Here are the birth-time indices for the types of length~3.
\[
\bt(ff) = (1,2,3),
\qquad
\bt(fg) = (3,1,2),
\qquad
\bt(gf) = (2,1,3),
\qquad
\bt(gg) = (3,2,1).
\]
\end{Example}

In summary, the information in a $\tau$-module~$\Vv$ which survives to time~$n$ is encoded as a filtration $R(\Vv)$ on~$V_n$. The `age' of the information at each level of the filtration (i.e.\ at each subquotient) is recorded in the birth-time index~$\bt(\tau)$.

For a simplified but precise version of this last claim, we now calculate the right-filtrations of interval $\tau$-modules. In the filtration specified in the following lemma, $J_i/J_{i-1} = \Ff$ is the only non-zero subquotient, corresponding to the birth time~$b_i$.

\begin{Lemma}
\label{lem:intervalformula}
Let $\tau$ be a type of length~$n$, with $\bt(\tau) = (b_1, b_2, \dots, b_n)$. For each $i = 1, 2,  \dots, n$, we have an isomorphism
\[
\rf(\Ii_\tau(b_i,n)) = \Jj(i,n),
\]
where $\Jj(i,n) = (J_0, J_1, \dots, J_n)$ is the filtration on~$\Ff$ defined by
\[
J_0 = \dots = J_{i-1} = 0; \quad
J_i = \dots = J_n = \Ff.
\]
\end{Lemma}

\begin{Remark}
We refer to the $\Jj(b,n)$ also as \emph{intervals} (but now in the category of filtered vector spaces).
\end{Remark}

\begin{proof}
This is a straightforward calculation by induction on~$\tau$. For the base case, $\tau$ is empty and $\bt(\tau) = (1)$. Then $\rf(\Ii(1,1)) = (0, \Ff) = \Jj(1,1)$ as claimed.
Now suppose the result is known for $\tau$, with $\bt(\tau) = (b_1, \dots, b_n)$. Suppose $\tau^+ = \tau f$ or~$\tau g$. In both cases, write $\bt(\tau^+) = (b_1^+, \dots, b_{n+1}^+)$.

Case~$f$: Suppose that $1 \leq i \leq n$; then $b_i^+ = b_i$ and therefore
\[
\Ii_{\tau^+}(b_i^+,n+1) = ( \Ii_\tau(b_i,n) \stackrel{1}{\rmap} \Ff).
\]
Writing $\rf(\Ii_\tau(b_i,n)) = \Jj(i,n) = (J_0, J_1, \dots, J_n)$,
it follows that
\[
\rf(\Ii_{\tau^+}(b_i^+,n+1)) = (J_0, J_1, \dots, J_n, \Ff) = \Jj(i,n+1).
\]
For $i=n+1$, we have $b_{n+1}^+ = n+1$, and indeed
\[
\rf(\Ii_{\tau^+}(n+1,n+1)) = \rf( (\dots) \stackrel{0}{\rmap} \Ff)
= (0, \dots, 0, \Ff) = \Jj(n+1,n+1).
\]

Case~$g$: Suppose that $2 \leq i \leq n+1$; then $b_i^+ = b_{i-1}$ and therefore
\[
\Ii_{\tau^+}(b_i^+,n+1) = ( \Ii_\tau(b_{i-1},n) \stackrel{1}{\lmap} \Ff).
\]
Writing $\rf(\Ii_\tau(b_{i-1},n)) = \Jj(i-1,n) = (J_0, J_1, \dots, J_n)$,
it follows that
\[
\rf(\Ii_{\tau^+}(b_i^+,n+1)) = (0, J_0, J_1, \dots, J_n) = \Jj(i,n+1).
\]
For $i=1$, we have $b_{1}^+ = n+1$ and then
\[
\rf(\Ii_{\tau^+}(n+1,n+1)) = \rf( (\dots) \stackrel{0}{\lmap} \Ff)
= (0, \Ff, \dots, \Ff) = \Jj(1,n+1)
\]
as required.
\end{proof}

Thus the right-filtration (with the help of the birth-time index) distinguishes the different intervals $\Ii(b,n)$. It gives no information about intervals $\Ii(b,d)$ when $d < n$, since in those cases $I_n = 0$.

\begin{Example}
Consider $\tau = fgf$, so $\bt(\tau) = (b_1, b_2, b_3, b_4) = (3,1,2,4)$ and in general
\[
\rf(\, V_1 \fmap{1} V_2 \gmap{2} V_3 \fmap{3} V_4 \,)
=
(0,\, f_3 g_2^{-1}(0),\, f_3 g_2^{-1} f_1(V_1),\,  f_3(V_3)  ,\, V_4).
\]
In particular,
\[
\begin{array}{cclccccclcrcc}
\Ii(b_2,4) &=&
\rf(\,
\Ff & \stackrel{1}{\rmap}
& \Ff & \stackrel{1}{\lmap}
& \Ff & \stackrel{1}{\rmap}
& \Ff \,)
& = & (0,\, 0,\, \Ff,\, \Ff,\, \Ff)
&=& \Jj(2,4)
\\
\Ii(b_3,4) &=&
\rf(\,
0 & \stackrel{}{\rmap}
& \Ff & \stackrel{1}{\lmap}
& \Ff & \stackrel{1}{\rmap}
& \Ff \,)
& = & (0,\, 0,\, 0,\, \Ff,\, \Ff)
&=& \Jj(3,4)
\\
\Ii(b_1,4) &=&
\rf(\,
0 & \stackrel{}{\rmap}
& 0 & \stackrel{}{\lmap}
& \Ff & \stackrel{1}{\rmap}
& \Ff \,)
& = & (0,\, \Ff,\, \Ff,\, \Ff,\, \Ff)
&=& \Jj(1,4)
\\
\Ii(b_4,4) &=&
\rf(\,
0 & \stackrel{}{\rmap}
& 0 & \stackrel{}{\lmap}
& 0 & \stackrel{}{\rmap}
& \Ff \,)
& = & (0,\, 0,\, 0,\, 0,\, \Ff)
&=& \Jj(4,4)
\\
\end{array}
\]
which is in accordance with Lemma~\ref{lem:intervalformula}.
\end{Example}

%--------------------------------------------------
\subsection{Decompositions of filtered vector spaces}
\label{subsec:filtvs}

We now consider filtered vector spaces in their own right, independently of the connection to zigzag-modules, and develop the theory of Remak decompositions. We will see later that this is the right tool for understanding Remak decompositions of zigzag modules.

A filtered vector space of depth~$n$ is a sequence $\Rr = (R_0, R_1, \dots, R_n)$ of vector spaces, where $0 = R_0 \leq R_1 \leq \dots \leq R_n$. 
The class of such objects is denoted $\filt_n$.
The right-filtration $\rf(\Vv)$ of any zigzag module~$\Vv$ of length~$n$ belongs to this class, as do the intervals $\Jj(i,n)$ defined in Lemma~\ref{lem:intervalformula}. Indeed, if $\Rr \in \filt_n$ satisfies $\dim(R_n)=1$, then $\Rr$ is isomorphic to some $\Jj(i,n)$.

\begin{Remark}
$\filt_n$ can be given the structure of a category in a natural way, but it is not quite an abelian category since morphisms do not generally have cokernels.
\end{Remark}

A filtered vector space $\Ss = (S_0, S_1, \dots, S_n)$ is a \textbf{subspace} of~$\Rr$ if $S_i \leq R_i$ for all~$i$.
It is appropriate to consider a stronger notion of subspace when dealing with direct-sum decompositions:
$\Ss$~%
%We say that $\Ss = (S_0, S_1, \dots, S_n)$
%
is an \textbf{induced subspace} of~$\Rr$ if there exists a vector subspace $K \leq R_n$ such that $S_i = R_i \cap K$ for all~$i$. In that event, we write $\Ss = \Rr \cap K$. Note that $K = R_n \cap K = S_n$.

We say that $\Rr$ is the \textbf{direct sum} of two subspaces, and write $\Rr = \Ss \oplus \Tt$, if $R_i = S_i \oplus T_i$ for all~$i$.
We claim that $\Ss, \Tt$ must be induced subspaces. Note that $S_n \cap T_n = 0$. For each~$i$, then, $R_i \cap S_n$ is a subspace of $R_i$ which contains $S_i$ and meets $T_i \leq T_n$ only at~0. It follows that $R_i \cap S_n = S_i$ for all~$i$. Thus $\Ss = \Rr \cap S_n$, and symmetrically $\Tt = \Rr \cap T_n$.

The general form of a direct-sum decomposition is therefore $\Rr = (\Rr \cap K) \oplus (\Rr \cap L)$. What are the requirements on $K, L$ to make this a valid decomposition? The direct sum condition implies that $R_n = K \oplus L$ as a vector space. Moreover, given a vector space decomposition $R_n = K \oplus L$, the further condition
\[
R_i = \Span\left( R_i \cap K, R_i \cap L \right)
\;
\mbox{for all~$i$}
\]
is necessary and sufficient to guarantee $\Rr = (\Rr \cap K) \oplus (\Rr \cap L)$.

If $\Rr = \Ss \oplus \Tt$, the two subspaces $\Ss, \Tt$ are said to be complementary summands. The following fact radically simplifies the decomposition theory of filtered vector spaces.

\begin{Proposition}
\label{prop:complement}
Every induced subspace of a filtered vector space has a complementary summand.
\end{Proposition}

\begin{proof}
We are given $\Ss = \Rr \cap K$, and seek to construct $\Tt = (T_0, T_1, \dots, T_n)$ such that $\Rr = \Ss \oplus \Tt$.
We proceed inductively.
Since $R_0 = S_0 = 0$ we take $T_0 = 0$. Now suppose we have chosen $T_k$ so that $R_k = S_k \oplus T_k$. In particular, $T_k \cap S_k = 0$. Then
\[
T_k \cap S_{k+1} \leq T_k \cap S_n 
= (T_k \cap R_k) \cap S_n
= T_k \cap (R_k \cap S_n)
= T_k \cap S_k = 0.
\]
Thus $T_k$ and $S_{k+1}$ are independent subspaces of $R_{k+1}$, and so $T_k$ can be extended to a complement $T_{k+1}$ of $S_{k+1}$ in $R_{k+1}$.
This completes the induction.
\end{proof}

\begin{Corollary}
The indecomposables in~$\filt_n$ are precisely the intervals $\Jj(i,n)$.
Thus, every filtered vector space can be decomposed as a finite direct sum of intervals.
\end{Corollary}

\begin{proof}
By Proposition~\ref{prop:complement}, $\Rr$ has nontrivial summands if and only if $R_n$ has nontrivial vector subspaces; this happens exactly when $\dim(R_n) > 1$.
\end{proof}

The dimension of $\Rr \in\filt_n$ is defined to be the vector of integers
\[
\dim(\Rr) = (c_1, c_2, \dots, c_n),
\]
where $c_i = \dim(R_i / R_{i-1})$ are the dimensions of the successive subquotients of the filtration.

\begin{Proposition}
\label{prop:filtinterval}
Let $\Rr$ be a filtered vector space of depth~$n$, with $\dim(\Rr) = (c_1, c_2, \dots, c_n)$. For any decomposition of $\Rr$ into intervals, the multiplicity of $\Jj(i,n)$ is~$c_i$. Thus:
\[
\Rr \cong \bigoplus_{1 \leq i \leq n} c_i\, \Jj(i,n).
\]
\end{Proposition}

\begin{proof}
Let $m_i$ be the multiplicity of $\Jj(i,n)$. Then, for all~$k$,
\[
\dim(R_k) = m_1 + m_2 + \dots + m_k
\]
by considering the contribution of each summand, whereas
\[
\dim(R_k) = c_1 + c_2 + \dots + c_k
\]
by considering the contribution of each subquotient $R_i/R_{i-1}$.
This is possible only if $m_i = c_i$ for all~$i$.
\end{proof}

This concludes our tour of the decomposition theory for filtered vector spaces. Now we must leverage this to achieve a decomposition theory for $\tau$-modules. In one direction, the relationship is straightforward:

\begin{Proposition}
\label{prop:directsum}
The right-filtration operation respects direct sums, in the sense that
\[
\rf(\Vv_1 \oplus \dots \oplus \Vv_N) = \rf(\Vv_1) \oplus \dots \oplus \rf(\Vv_N)
\]
for $\tau$-modules $\Vv_1, \dots, \Vv_N$.
\end{Proposition}

\begin{proof}
This is proved by induction on~$\tau$, following the recursive structure of Definition~\ref{def:rfilt} and using the standard facts
\[
(f_1 \oplus \dots \oplus f_N)(R_1 \oplus \dots \oplus R_N)
= f_1(R_1) \oplus \dots \oplus f_N(R_N)
\]
and
\[
(g_1 \oplus \dots \oplus g_N)^{-1}(R_1 \oplus \dots \oplus R_N)
= g_1^{-1}(R_1) \oplus \dots \oplus g_N^{-1}(R_N)
\]
from linear algebra. (For simplicity we are suppressing various indices here.)
\end{proof}

However, what we need is a converse to Proposition~\ref{prop:directsum}: if the filtered vector space $\Rr = \rf(\Vv)$ can be split as a direct sum $\Rr = \Rr_1 \oplus \dots \oplus \Rr_N$, we would like to infer a corresponding splitting $\Vv = \Vv_1 \oplus \dots \oplus \Vv_N$ of $\tau$-modules.
In the following two sections we establish such a principle for a particular class: the `streamlined' $\tau$-modules.

%--------------------------------------------------
\subsection{Streamlined modules}
\label{subsec:streamlined}

We introduce a special class of $\tau$-module for which the right-filtration functor preserves all structural information.

\begin{Definition}
A $\tau$-module~$\Vv$ is \textbf{(right-)streamlined} if each $\stackrel{f_i}{\rmap}$ is injective and each $\stackrel{g_i}{\lmap}$ is surjective. 
\end{Definition}

Similarly, we may say that a $\tau$-module~$\Vv$ is left-streamlined if each $\stackrel{f_i}{\rmap}$ is surjective and each $\stackrel{g_i}{\lmap}$ is injective. We will not need to consider left-streamlined modules until Section~\ref{sec:further}. By default, `streamlined' will be taken to mean `right-streamlined'.

\begin{Example}
\label{ex:streaminterval}
Intervals $\Ii(b,n)$ are streamlined (but not $\Ii(b,d)$ for $d < n$). Conversely, a streamlined $\tau$-module $\Vv$ with $\dim(V_n) = 1$ is necessarily isomorphic to some $\Ii(b,n)$. Indeed, $\dim(V_i)$ is a non-decreasing sequence and therefore comprises some $b-1$ zeros (where $1 \leq b \leq n$) followed by $n-b+1$ ones. The maps between the one-dimensional terms are injective or surjective, and therefore isomorphisms.
\end{Example}

\begin{Proposition}
\label{prop:splitstream}
A direct sum $\Vv = \Vv_1 \oplus \dots \oplus \Vv_N$ of $\tau$-modules is streamlined if and only if each summand is streamlined.
\end{Proposition}

\begin{proof}
Each $\fmap{}$ in~$\Vv$ decomposes as $f = f_1 \oplus \dots \oplus f_N$ and is injective if and only if each $f_j$ is injective. Each $\gmap{}$ in~$\Vv$ decomposes as $g = g_1 \oplus \dots \oplus g_N$ and is surjective if and only if each $g_j$ is surjective.
\end{proof}

The proof of the following lemma appears at the end of this section.

\begin{Lemma}[Decomposition Lemma]
\label{lem:decomp}
Let $\Vv$ be a streamlined $\tau$-module and let $\Rr = \rf(\Vv)$. For any decomposition $\Rr = \Ss_1 \oplus \dots \oplus \Ss_N$, there exists a unique decomposition $\Vv = \Ww_1 \oplus \dots \oplus \Ww_N$ such that $\Ss_i = \rf(\Ww_j)$ for all~$j$.
\end{Lemma}

\begin{Theorem}[Interval decomposition for streamlined modules]
\label{thm:streamlinedremak}
Let $\Vv$ be a streamlined $\tau$-module of length~$n$, and write
$\dim(\rf(\Vv)) = (c_1, c_2, \dots, c_n)$ and $\bt(\tau) = (b_1, b_2, \dots, b_n)$.
Then there is an isomorphism of $\tau$-modules
\[
\Vv \cong \bigoplus_{1 \leq i \leq n} c_i \, \Ii(b_i, n).
\]
\end{Theorem}

\begin{proof}
Let $\Rr = \rf(\Vv)$. By Proposition~\ref{prop:filtinterval}, there is a decomposition $\Rr = \Jj_1 \oplus \dots \oplus \Jj_N$, where the $\Jj_j$ are a collection of $N = c_1 + \dots + c_n$ intervals with $\Jj(i,n)$ occuring with multiplicity~$c_i$. Lemma~\ref{lem:decomp} produces a decomposition $\Vv = \Ii_1 \oplus \dots \oplus \Ii_N$, with $\rf(\Ii_j) = \Jj_j$ for all~$j$. Each $\Ii_j$ is streamlined (Proposition~\ref{prop:splitstream}) with maximum dimension $\dim((I_j)_n) = 1$, and is therefore isomorphic to some $\Ii(b,n)$ (Example~\ref{ex:streaminterval}). By Lemma~\ref{lem:intervalformula}, we must have $\Ii_j = \Ii(b_i,n)$ whenever $\Jj_j = \Jj(i,n)$. It follows that the $\Ii_j$ are a collection of $N = c_1 + \dots + c_n$ intervals with $\Ii(b_i,n)$ occuring with multiplicity~$c_i$.
\end{proof}

We complete this chapter with a proof of the Decomposition Lemma.

\begin{proof}[Proof of Lemma~\ref{lem:decomp}]
We may assume that $N=2$, since the general case follows by iteration. Accordingly, suppose that $\Rr = \rf(\Vv)$ can be written in the form $\Rr = \Ss \oplus \Tt$; we must show that there is a corresponding decomposition $\Vv =\Ww \oplus \Xx$. We will argue by induction on $n = \operatorname{\rm length}(\tau)$.

The first step is to determine the splitting $V_n = W_n \oplus X_n$. In fact, the stipulation that $\Ss = \rf(\Ww)$ and $\Tt = \rf(\Xx)$ forces $W_n = S_n$ and $X_n = T_n$. If $n=1$, then we are done. 
Otherwise, let $\hat\Vv$ denote the truncation of $\Vv$ to the indices $\{1, \dots, n-1\}$ and let $\hat\Rr = \rf(\hat\Vv)$. We will shortly establish that $\Rr = \Ss \oplus \Tt$ induces a unique compatible decomposition $\hat\Rr = \hat\Ss \oplus \hat\Tt$. The inductive hypothesis will then provide $\hat\Vv = \hat\Ww \oplus \hat\Xx$, which combines with $V_n = W_n \oplus X_n$ to produce the desired decomposition $\Vv = \Ww \oplus \Xx$. That will complete the proof.

Write $\Rr = (R_0, R_1, \dots, R_n)$. There are two cases.

Case~$\fmap{n-1}$, injective. We can identify $V_{n-1}$ with the subspace $f_{n-1}(V_{n-1}) = R_{n-1}$ of $V_n$. Thereupon we have
\[
\hat\Rr = (R_0, R_1, \dots, R_{n-1}).
\]
The unique splitting of $V_{n-1}$ compatible with $V_n = W_n \oplus X_n$ is
\[
V_{n-1} = (R_{n-1} \cap W_n) \oplus (R_{n-1} \cap X_n) = S_{n-1} \oplus T_{n-1}.
\]
We must now verify that the induced subspaces $\hat\Ss = \hat\Rr \cap S_{n-1}$ and $\hat\Tt = \hat\Rr \cap T_{n-1}$ give a valid decomposition $\hat\Rr = \hat\Ss \oplus \hat\Tt$ of filtered vector spaces. This follows because $\hat S_{i} = R_i \cap S_{n-1} = R_i \cap S_n = S_i$ and similarly $\hat T_{i} = T_i$, for all~$i < n$; so $R_i = S_i \oplus T_i = \hat S_i \oplus \hat T_i$ as required.

Case~$\gmap{n_1}$, surjective. Here we identify $V_{n-1}$ as the quotient $V_n / \ker(g_{n-1}) = R_n / R_1$. Under this identification,
\[
\hat\Rr = (R_1/R_1, R_2/R_1, \dots, R_n/R_1).
\]
In splitting $V_{n-1} = W_{n-1} \oplus X_{n-1}$ we are compelled to take
\[
W_{n-1} = g_{n-1}(W_n) = S_n / S_1,
\qquad
X_{n-1} = g_{n-1}(X_n) = T_n / T_1,
\]
which induce
\[
\hat{S}_{i} = g_{n-1}(S_{i+1}) = S_{i+1} / S_1,
\qquad
\hat{T}_{i} = g_{n-1}(T_{i+1}) = T_{i+1}/ T_1,
\]
for the purported splitting $\hat\Rr = \hat\Ss \oplus \hat\Tt$.
To confirm that this is a genuine decomposition, we note from linear algebra that the twin facts
\[
R_{i+1} = S_{i+1} \oplus T_{i+1},
\qquad
R_1 = S_1 \oplus T_1 = (S_{i+1} \cap R_1) \oplus (T_{i+1} \cap R_1)
\]
imply that
\[
R_{i+1}/R_1 = (S_{i+1}/S_1) \oplus (T_{i+1}/T_1)  
\]
as required.
\end{proof}

\begin{Remark}
There is a high-level proof of Lemma~\ref{lem:decomp} which in some sense is the natural explanation for the result. We outline this proof now. The first observation is that the transformation $\Vv \to \rf(\Vv)$ is a functor from $\taumod$ to $\filt_n$: a morphism $\alpha: \Vv \to \Ww$ induces a morphism $\rf(\alpha): \rf(\Vv) \to \rf(\Ww)$. Indeed, $\rf(\alpha)$ is defined to be $\alpha_n: V_n \to W_n$; one must check that this respects the filtrations on $V_n$ and $W_n$.
Being a functor, $\rf$ defines a ring homomorphism $\End(\Vv) \to \End(\rf(\Vv))$. The second key fact is that this homomorphism is an isomorphism if $\Vv$ is streamlined (in general it is surjective). It is well known that direct-sum decompositions of a module can be extracted from the structure of its endomorphism ring: direct summands correspond to idempotent elements of the ring. It follows that $\Vv$ and $\rf(\Vv)$ have the same decomposition structure.
\end{Remark}

%------------------------------------------------------------------
\section{The Interval Decomposition Algorithm}
\label{sec:algorithms}

Here we describe the algorithm for determining the indecomposable factors of a $\tau$-module. We give three versions of the `algorithm'.

The first version, in Section~\ref{subsec:maintheorem}, is not an algorithm but a proof that every $\tau$-module decomposes as a sum of interval modules (Theorem~\ref{thm:interval}). Moreover, the structure of the proof makes it clear how to compute the interval decomposition (Theorem~\ref{thm:main}). The algorithms in the subsequent sections build on this.

In Section~\ref{subsec:algabst} we describe an abstract form of the decomposition algorithm, using the language of vector spaces and linear maps. No consideration is given to how the spaces and maps are described and manipulated in practice.

In Section~\ref{subsec:algconc} we suppose that the maps $f_i, g_i$ are presented concretely as matrices $M_i, N_i$ with respect to a choice of bases for the vector spaces~$V_i$. We describe an algorithm which takes these matrices as input and returns the interval decomposition.

%--------------------------------------------------
\subsection{The interval decomposition theorem}
\label{subsec:maintheorem}

Our present goal is to give a somewhat constructive proof of Theorem~\ref{thm:interval}, which asserts that any $\tau$-module~$\Vv$ is isomorphic to a direct sum of intervals $\Ii(b,d)$. We prove a stronger, more precise result, which explicitly determines the multiplicity of each interval within~$\Vv$.

Some notation will help with the theorem statement. If
\[
\Vv =  (V_1 \pmap{1} \dots \pmap{n-1} V_n)
\]
then let
\[
\Vv[k] =  (V_1 \pmap{1} \dots \pmap{k-1} V_k)
\]
denote the truncation of $\Vv$ to length~$k$, and let $\tau[k]$ denote its type (which is a truncation of~$\tau$).

\begin{Theorem}[Interval decomposition]
\label{thm:main}
Let $\Vv$ be a $\tau$-module. For $1 \leq k \leq n$, define
\[
(b_1^k, b_2^k, \dots, b_k^k) = \bt(\tau[k]).
\]
Writing $\Rr_k = \rf(\Vv[k])$, define
\[
(c_1^k, c_2^k, \dots, c_k^k) = 
\left\{
\begin{array}{l}
\dim(\Rr_k \cap \Ker(f_k) )
\\
%\dim(\Rr_k / \Img(g_k)) =
\dim(\Rr_k) - \dim(\Rr_k \cap \Img(g_k))
%\dim(\rf(\Vv[k])) - \dim(\rf(\Vv[k]) \cap \Img(g_k))
\end{array}
\right.
\]
(whichever is applicable) when $k \ne n$, and
\[
(c_1^n, c_2^n, \dots, c_n^n) = \dim(\Rr_n).
\] 
Then
\[
\Vv \cong \bigoplus_{1 \leq i \leq k \leq n} c_i^k\, \Ii(b_i^k, k).
\]
\end{Theorem}

\begin{Addendum}
\label{add:main}
In the situation of Theorem~\ref{thm:main}, write
\[
(r^k_1, \dots, r^k_k) = \dim(\Rr_k)
\]
for $k = 1,\dots, n$, and conventionally define $r^{n+1}_i = 0$ for all~$i$. Then
\[
c^k_i = \left\{
\begin{array}{ll}
r^k_i - r^{k+1}_i &\qquad \mbox{\rm case~$\fmap{k}$} \\
r^k_i - r^{k+1}_{i+1} &\qquad \mbox{\rm case~$\gmap{k}$}
\end{array}
\right.
\]
for $1 \leq i \leq k \leq n$.
\end{Addendum}

The decomposition strategy begins with the following lemma. The idea is to proceed from left to right along the complex, removing streamlined summands at each step. Having done this, the Remak decompositions of those summands can be determined by counting dimensions, as prescribed in Theorem~\ref{thm:streamlinedremak}.

\begin{Lemma}
\label{lem:cascade}
Let $\Vv = V_1 \pmap{1} \dots \pmap{n-1} V_n $ be an irreducible $\tau$-module of length~$n$. Then there exists a direct-sum decomposition
\[
\Vv = \Vv^1 \oplus \Vv^2 \oplus \dots \oplus \Vv^n
\]
where each $\Vv^k$ is supported over the indices $\{1, 2, \dots, k\}$
and is right-streamlined over that range.
\end{Lemma}

The following picture illustrates the decomposition.
\[
\Vv = 
\left\{
\begin{array}{ccccccccccc}
\Vv^1 &=& V^1_1 
\\\oplus\\
\Vv^2 &=& V^2_1 & \pmap{1} & V^2_2
\\\oplus\\
\Vv^3 &=& V^3_1 & \pmap{1} & V^3_2 &\pmap{2} & V^3_3
\\\oplus\\
\vdots
\\\oplus\\
\Vv^n &=& V^n_1 & \pmap{1} & V^n_2 &\pmap{2} & V^n_3
& \pmap{3} & \cdots & \pmap{n-1} & V^n_n
\end{array}
\right.
\]
Each row (i.e.\ summand) is right-streamlined, and therefore amenable to analysis via the right-filtration functor.

\begin{proof}
We proceed by induction on the length of~$\Vv$.
%Let $\Vv[k]$ denote the truncation of~$\Vv$ to length~$k$ (i.e. over the indices $\{1, \dots, k\}$).
%
The inductive statement is that
\[
\Vv[k] = \Vv^1 \oplus \dots \oplus \Vv^{k-1} \oplus \Ww^k
\]
where the $\Vv^i$ are as in the theorem statement, and $\Ww^k$ is itself right-streamlined.

For the base case $k=1$, there is nothing to prove: take $\Ww^1 = \Vv[1]$. Now suppose the inductive statement is established for~$k$, and consider $\Vv[k+1]$. This can be written
\begin{eqnarray*}
V[k+1] &=& (\Vv^1 \oplus \dots \oplus \Vv^{k-1} \oplus \Ww^k) \pmap{k} V_{k+1}
\\
&=& \Vv^1 \oplus \dots \oplus \Vv^{k-1} \oplus (\Ww^k \pmap{k} V_{k+1})
\end{eqnarray*}
where the rebracketing is permissible because all of the $\Vv^i$ terms terminate before time~$k$, and therefore do not interact with $\pmap{k}$.
The goal now is to rewrite $(\Ww^k \pmap{k} V_{k+1})$ as $\Vv^k \oplus \Ww^{k+1}$, where $\Vv^k$ terminates at time~$k$ and both $\Vv^k$ and $\Ww^{k+1}$ are right-streamlined. The rightmost term of $\Ww^k$ is $V_k$, so $\rf(\Ww^k)$ is a filtration on~$V_k$.

Case~$f$: $\Ww^k \fmap{k} V_{k+1}$. In other words $f_k: V_k \to V_{k+1}$. Let $\Ss = \rf(\Ww^k) \cap \Ker(f_k)$. Proposition~\ref{prop:complement} implies that $\Ss$ has a complement in~$\rf(\Ww^k)$; say $\rf(\Ww^k) = \Ss \oplus \Tt$. This corresponds (Lemma~\ref{lem:decomp}) to a direct sum decomposition $\Ww^k = \Vv^k \oplus \hat\Ww^{k}$, where both summands are streamlined (Proposition~\ref{prop:splitstream}). This defines $\Vv^k$, and we set $\Ww^{k+1} = (\hat\Ww^{k} \fmap{k} V_{k+1})$.
To check that this works, note that $f_k$ is zero on $(\Vv^k)_k = \Ker(f_k)$ and is injective on the complementary subspace $(\hat\Ww^k)_k$. Thus $\Vv^k$ is a summand of $\Vv[k+1]$ terminating at time~$k$, and $\Ww^{k+1}$ is streamlined.

Case~$g$: $\Ww^k \gmap{k} V_{k+1}$. In other words $g_k: V_{k+1} \to V_{k}$. Let $\Ss = \rf(\Ww^k) \cap \Img(g_k)$. Proposition~\ref{prop:complement} implies that $\Ss$ has a complement in~$\rf(\Ww^k)$; say $\rf(\Ww^k) = \Ss \oplus \Tt$. This corresponds (Lemma~\ref{lem:decomp}) to a direct sum decomposition $\Ww^k = \hat\Ww^{k} \oplus \Vv^k$, where both summands are streamlined (Proposition~\ref{prop:splitstream}). This defines $\Vv^k$, and we set $\Ww^{k+1} = (\hat\Ww^{k} \gmap{k} V_{k+1})$.
To check that this works, note that $g_k$ is surjective onto $(\hat\Ww^k)_k = \Img(g_k)$ and misses the complementary subspace $(\Vv^k)_k$. Thus $\Vv^k$ is a summand of $\Vv[k+1]$ terminating at time~$k$, and $\Ww^{k+1}$ is streamlined.

This establishes the inductive step, so eventually
\[
\Vv = \Vv[n] = \Vv^1 \oplus \dots \oplus \Vv^{n-1} \oplus \Ww^n
\]
and we set $\Vv^n = \Ww^n$ to finish the proof.
\end{proof}

\begin{proof}[Proof of Theorem~\ref{thm:main}]
Write $\Vv = \Vv^{1} \oplus \dots \oplus \Vv^n$ according to Lemma~\ref{lem:cascade}. We now calculate the decomposition of each $\Vv^k$ into intervals $\Ii(b,k)$. Note that
\[
\Vv[k] = \Vv^k \oplus \Vv^{k+1}[k] \oplus \dots \oplus \Vv^n[k].
\]
We can write $\Ww^k = \Vv^{k+1} \oplus \dots \oplus \Vv^n$, so then
\[
\Rr_k = \rf(\Vv^k \oplus \Ww^k[k]) = \rf(\Vv^k) \oplus \rf(\Ww^k[k])
\]
(using Proposition~\ref{prop:directsum}). This is a filtration on $V^k_k \oplus W^k_k$.

Suppose $k < n$. We note that $\Ww^k$ is streamlined up to time $k+1$, whereas $\Vv^k$ is zero at time $k+1$.
The next map in the sequence is
\[
V^k_k \oplus W^k_k \fmap{k} W^k_{k+1}
\qquad
\mbox{or}
\qquad
V^k_k \oplus W^k_k \gmap{k} W^k_{k+1}.
\]
In the first case, it follows that $V^k_k = \Ker(f_k)$ and therefore $\rf(\Vv^k) = \Rr_k \cap \Ker(f_k)$. In the second case, $V^k_k$ is a complement to $\Img(g_k)$ in $V_k$, so
$\Rr_k = \rf(\Vv^k) \oplus (\Rr_k \cap \Img(g_k))$.
Thus
\[
\dim(\rf(\Vv^k)) = 
\left\{
\begin{array}{l}
\dim(\Rr_k \cap \Ker(f_k) )
\\
%\dim(\Rr_k / \Img(g_k)) =
\dim(\Rr_k) - \dim(\Rr_k \cap \Img(g_k))
%\dim(\rf(\Vv[k])) - \dim(\rf(\Vv[k]) \cap \Img(g_k))
\end{array}
\right\}
=
(c^k_1, \dots, c^k_k)
\]
(whichever middle term is applicable). When $k=n$, moreover, we have
\[
\dim(\rf(\Vv^n)) = \dim(\Rr_n) = (c^n_1, \dots, c^n_n).
\]
Thus, at last,
\[
\Vv = \bigoplus_{1 \leq k \leq n} \Vv^k
\cong
  \bigoplus_{1 \leq k \leq n}
  \left\{
  \bigoplus_{1 \leq i \leq k} c_i^k\, \Ii(b_i^k, k)
  \right\}
\]
using Theorem~\ref{thm:streamlinedremak} to decompose the $\Vv^k$.
\end{proof}

\begin{proof}[Proof of Addendum~\ref{add:main}]
Write $(w^k_1, \dots, w^k_k) = \dim(\rf(\Ww^k[k]))$. Since $\Rr_k = \rf(\Vv^k) \oplus \rf(\Ww^k[k])$ we can take dimensions and obtain the formula
\[
(r^k_1, \dots, r^k_k) = (c^k_1, \dots, c^k_k) + (w^k_1, \dots, w^k_k).
\]
Note also that $\Rr_{k+1} = \rf(\Vv[k+1]) = \rf(\Ww^k[k+1])$. Moreover, $\Ww^k$ is streamlined up to time~$k+1$.
It follows that
\[
(r^{k+1}_1, \dots, r^{k+1}_{k+1})
= \dim(\rf(\Ww^k[k+1]))
= 
\left\{
\begin{array}{ll}
(w^k_1, \dots, w^k_k, ?) &\qquad \mbox{case~$f$}
\\
(?, w^k_1, \dots, w^k_k) &\qquad \mbox{case~$g$}
\end{array}
\right.
\]
and therefore
\[
c^k_i = r^k_i - w^k_i = 
\left\{
\begin{array}{ll}
r^k_i - r^{k+1}_i &\qquad \mbox{case~$f$} \\
r^k_i - r^{k+1}_{i+1} &\qquad \mbox{case~$g$}
\end{array}
\right.
\]
which is the desired formula.
\end{proof}

%--------------------------------------------------
\subsection{Abstract vector spaces}
\label{subsec:algabst}

We now transcribe Theorem~\ref{thm:main} as an abstract algorithm for determining the interval structure of a $\tau$-module $\Vv$ of length~$n$. This algorithm will serve as a skeleton for the more concrete algorithms developed later.

\begin{Algorithm}
\label{alg:main}
We proceed through $k = 1, 2, \dots, n$, computing the filtration $\Rr_k = \rf(\Vv[k])$, the birth-time index $\bt(\tau[k])$, and the dimensions $c^k_i$ iteratively.

\textsc{begin}

\textbf{\small Initialisation} ($k=1$): 
\begin{itemize}
\item[(1)]
%$\rf(\Vv[1]) = (0, V_1)$.
$\Rr_1 = (0, V_1)$.
\item[(2)]
$\bt(\tau[1]) = (1)$.
\end{itemize}

\textbf{\small Iterative step} ($k = 1, 2, \dots, n-1$):
\begin{itemize}
\item[(3)]
Calculate $\Rr_{k+1}$ from $\Rr_k = (R^k_0, R^k_1, \dots, R^k_k)$
using Definition~\ref{def:rfilt}:
\[
%\bt(\tau[k+1])) = 
(R^{k+1}_0, R^{k+1}_1, \dots, R^{k+1}_{k+1}) = 
\left\{
\begin{array}{ll}
(f_k(R^k_0), f_k(R^k_1), \dots, f_k(R^k_k), V_{k+1})
 &\qquad\mbox{case~$f$}
\\
(0, g_k^{-1}(R^k_0), g_k^{-1}(R^k_1), \dots, g_k^{-1}(R^k_k))
 &\qquad\mbox{case~$g$}
\end{array}
\right.
\]

\item[(4)]
Calculate $\bt(\tau[k+1]))$ from $\bt(\tau[k]) = (b^k_1, b^k_2, \dots, b^k_k)$
using Definition~\ref{def:birthtime}:
\[
%\bt(\tau[k+1])) = 
(b^{k+1}_1, \dots, b^{k+1}_{k+1}) = 
\left\{
\begin{array}{ll}
(b^k_1, \dots, b^k_k, k+1) &\qquad\mbox{case~$f$}
\\
(k+1, b^k_1, \dots, b^k_k) &\qquad\mbox{case~$g$}
\end{array}
\right.
\]

\item[(5)]
Calculate $(c^k_1, \dots, c^k_k)$ using the formula in Theorem~\ref{thm:main}:
\[
(c_1^k, c_2^k, \dots, c_k^k) = 
\left\{
\begin{array}{ll}
\dim(\Rr_k \cap \Ker(f_k) )
 &\qquad\mbox{case~$f$}
\\
\dim(\Rr_k) - \dim(\Rr_k \cap \Img(g_k))
 &\qquad\mbox{case~$g$}
\end{array}
\right.
\]
Alternatively, use the formula in Addendum~\ref{add:main}:
\[
c^k_i =
\left\{
\begin{array}{ll}
r^k_i - r^{k+1}_{i}
 &\qquad\mbox{case~$f$}
\\
r^k_i - r^{k+1}_{i+1}
 &\qquad\mbox{case~$g$}
\end{array}
\right.
\]
Here $(r^k_1, \dots, r^k_k) = \dim(\Rr_k)$.
\end{itemize}

\textbf{\small Terminal step} ($k=n$):
\begin{itemize}
\item[(6)]
Calculate $(c^n_1, \dots, c^n_n) = \dim(\Rr(\Vv))$.
\end{itemize}

\textbf{\small Print results}:
\begin{itemize}
\item[(7)]
For $1 \leq i \leq k \leq n$, the interval $\Ii(b^k_i, k)$ occurs with multiplicity~$c^k_i$.
\end{itemize}

\textsc{end}

\end{Algorithm}

Note that steps (3--5) have an `$f$' verson and a `$g$' version, depending on the direction of the map~$p_k$.

This abstract algorithm does not specify how the filtered vector spaces $\rf(\Vv[k+1])$ are stored, nor how the maps $f_k$ or~$g_k$ (which are used in steps (3) and~(5)) are represented.
In any concrete setting, it is necessary to specify data structures. A good choice will facilitate the calculations in steps (3) and~(5). In the next section, we work out the details in a simple scenario.

%--------------------------------------------------
\subsection{Concrete vector spaces}
\label{subsec:algconc}

In this section we describe an algorithm to solve the following concrete problem. Let $\tau$ be a type of length~$n$. We specify a $\tau$-module~$\Vv$ as follows. Set $V_i = \Ff^{a_i}$ for integers $a_i \geq 0$. For each~$i$, the map $f_i$ is defined by an $a_{i+1}$-by-$a_i$ matrix $M_i$ or else the map $g_i$ is defined by an $a_{i}$-by-$a_{i+1}$ matrix $N_{i}$. We are to determine $\Pers(\Vv)$, given $\tau$ and the matrices $M_i$ or $N_i$.

We follow Algorithm~\ref{alg:main}. The substantial task is to calculate the sequence of right-filtrations $\Rr_k = \rf(\Vv[k])$, for step~(3). Everything else is book-keeping: the birth-time indices~$b_i^k$ are calculated according to step~(4); and the filtration dimensions~$r^k_i$ (and hence the~$c_i^k$) will be easy to read off from the stored description of the filtrations.

\subsubsection*{Basis transformations}
The algorithm operates on two levels. On the \emph{conceptual} level, we proceed by modifying the bases of the spaces~$V_i$ by elementary basis transformations. Initially each basis $\Bb_i$ is the standard basis of~$\Ff^{a_i}$. We perform modifications on $\Bb_2, \Bb_3, \dots, \Bb_{n-1}$ in sequence. On the \emph{pragmatic} level, 
what we actually do is apply elementary row and column operations to the matrices $M_i$ or $N_i$. We make no attempt to track the bases themselves; instead we implement the effect of those changes on the matrices.

Suppose we apply elementary basis transformations to $\Bb_{k+1}$ on the conceptual level. On the pragmatic level, we must perform
\[
\mbox{row operations on $M_{k}$}
\qquad
\mbox{or}
\qquad
\mbox{column operations on $N_{k}$}
\]
and simultaneously perform
\[
\mbox{column operations on $M_{k+1}$}
\qquad
\mbox{or}
\qquad
\mbox{row operations on $N_{k+1}$}
\]
to enact those transformations. Thus, at every stage we must make parallel changes to two matrices simultaneously. Usually we are working to put $M_{k}$ or $N_{k}$ in a particular form, and while doing so the changes have to be mirrored in $M_{k+1}$ or $N_{k+1}$ (paying no attention yet to the structure of that matrix).

We now make this precise. The \textbf{elementary transformation} $E_i(p,q,\lambda)$ is defined as follows. On the conceptual level, this is a modification of $\Bb_i = (\beta_1, \dots, \beta_{a_i})$ involving basis vectors $\beta_p$ and~$\beta_q$:
\begin{eqnarray*}
\beta_p &\leftarrow& \beta_p\\ 
\beta_q &\leftarrow& \beta_q + \lambda \beta_p
\end{eqnarray*}
On the pragmatic level, if $L$ is a matrix representing a linear map $V_i \to W$ for some~$W$ (this will be $N_{i-1}$ or $M_i$ in our situation), then we modify the columns of~$L$ accordingly:
\begin{eqnarray*}
\col_p &\leftarrow& \col_p\\ 
\col_q &\leftarrow& \col_q + \lambda \col_p
\end{eqnarray*}
Else, if $L$ represents a linear map of the form $W \to V_i$ (this will be $M_{i-1}$ or $N_i$ in our situation) then we must apply the dual transformation to the rows of~$L$:
\begin{eqnarray*}
\row_p &\leftarrow& \row_p - \lambda \row_q \\ 
\row_q &\leftarrow& \row_q
\end{eqnarray*}
In spirit, we right-multiply by the matrix
$\left[ \begin{array}{rr} 1 & \lambda \\ 0 & 1\end{array} \right]$
to modify columns, or else left-multiply by the inverse matrix
$\left[ \begin{array}{rr} 1 & -\lambda \\ 0 & 1\end{array} \right]$
to modify rows.

Besides the elementary transformations $E_i(p,q,\lambda)$, it is sometimes appropriate to permute the basis elements. The operation $P_i(p,q)$ of interchanging $\beta_p$ with~$\beta_q$ is realised pragmatically by interchanging $\col_p$ with~$\col_q$, or $\row_p$ with~$\row_q$, as appropriate. %There are no surprises here.

\subsubsection*{Filtrations}

The filtration $\Rr_k = \rf(\Vv[k])$ on~$V_k$ is to be represented as follows. We require the basis $\Bb_k = (\beta_1, \dots, \beta_{a_i})$ to be compatible with the filtration, in a sense that will become clear. Assuming such a basis, the filtration $\Rr_k = (R_0, R_1, \dots, R_k)$ is represented as a non-decreasing function
\[
\phi_k: \{ 1, 2, \dots, a_i \} \to \{1, \dots, k\}
\]
so that
\[
R_i = \Span \left\{ \beta_p \mid \phi_k(p) \leq i  \right\}
\]
for $i = 1, \dots, k$. In other words: the first few basis elements (those $\beta_p$ with $\phi_k(p)=1$) form a basis for~$R_1$; the next few basis elements extend this to a basis for~$R_2$, and so on.
The dimension $r_i^k = \dim(R_i/R_{i-1})$ can be read off as the cardinality of $\phi_k^{-1}(i)$.

\subsubsection*{Gaussian elimination}

Step~(3) boils down to the following task. Suppose that $\Bb_k$ and $\phi_k$ together represent the filtration $\Rr_k$; then modify $\Bb_{k+1}$ and determine $\phi_{k+1}$ to represent $\Rr_{k+1}$.
We now explain how to do this.

%We achieve this pragmatically by applying a `filtered' Gaussian elimination to the matrix $M_k$ or~$N_k$.
%%
%Gaussian elimination implicitly entails changes of basis (of the column space, if there are row operations; or of the row space, if there are column operations). We must be careful that $\Bb_k$ remains compatible with the filtration throughout the procedure. Accordingly:
%\begin{itemize}
%\item
%$E_k(p,q,\lambda)$ is permissible when $\phi_k(p) \leq \phi_k(q)$.
%\item
%$P_k(p, q)$ is permissible when $\phi_k(p) = \phi_k(q)$.
%\end{itemize}
%These rules guarantee that $\Span\{ \beta_p \mid \phi_k(p) \leq i \}$ is unchanged.

Case~$M$: the matrix $M_k$ represents a linear map $V_k \to V_{k+1}$. We assume that $\Bb_k$ is compatible with the filtration~$\Rr_k$, and that $\phi_k$ identifies the filtration. This gives a block structure
\[
M_k = \left[ 
\begin{array}{cccc}
K_1 & K_2 & \cdots & K_k
\end{array}
\right]
\]
where $K_i$ gathers together the columns~$q$ with $\phi_k(q)=i$.
Using row operations only, put $M_k$ into (unreduced) row echelon form. This means:
\begin{itemize}
\item
Each of the top $r$ rows contains a 1 (the \emph{pivot}) as its leftmost nonzero entry.
\item
Each pivot lies strictly to the left of the pivots of the rows below it.
\item
The lowest $a_{k+1}-r$ rows are entirely zero.
\end{itemize}
These row operations correspond to elementary operations $E_{k+1}(p,q,\lambda)$, and the effect of these operations is felt on the next matrix $M_{k+1}$ or~$N_{k+1}$, which must be modified accordingly.
We now define $\phi_{k+1}$ as follows:
\[
\phi_{k+1}(p) = 
\left\{
\begin{array}{ll}
\phi_k(q) & \qquad \mbox{if row~$p$ has a pivot in column~$q$,}
\\
k+1 &  \qquad \mbox{if row~$p$ has no pivot.}
\end{array}
\right.
\]
See Figure~\ref{fig:gaussM}.
\begin{figure}
\centerline{
\includegraphics[scale=.5]{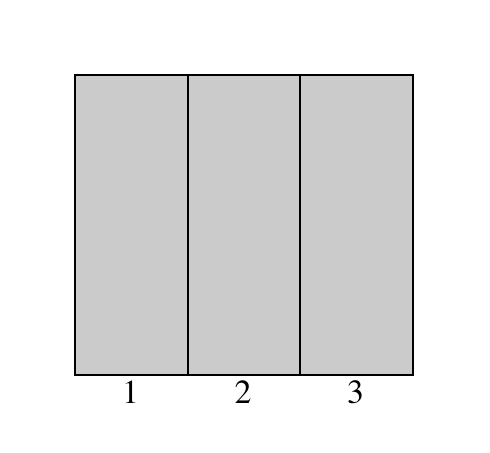}
\qquad
\includegraphics[scale=.5]{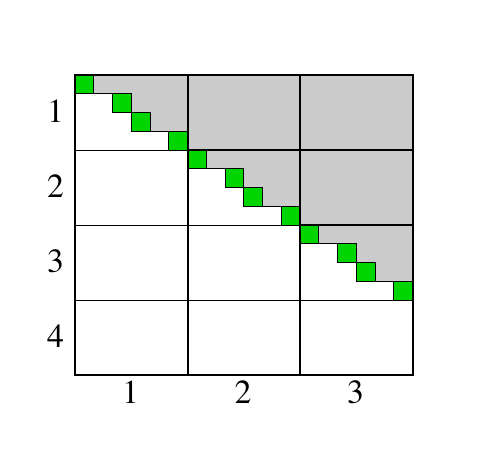}
}
\caption{Using row echelon form to compute $\Rr_{k+1}$ from $\Rr_k$.}
\label{fig:gaussM}
\end{figure}
It is evident in the figure that $R^k_i$ maps onto $R^{k+1}_i$ for all~$i$.

Case~$N$: the matrix $N_k$ represents a linear map $V_{k+1} \to V_{k}$. We assume that $\Bb_k$ is compatible with the filtration~$\Rr_k$, and that $\phi_i$ identifies the filtration. This time we have a vertical block structure
\[
N_k = \left[ 
\begin{array}{c}
L_1 \\ L_2 \\ \vdots \\ L_k
\end{array}
\right]
\]
where $L_i$ gathers together the rows~$q$ with $\phi_k(q) = i$.
Using column operations only, put $N_k$ into the column echelon form defined as follows (this echelon form begins on the bottom left):
\begin{itemize}
\item
Each of the leftmost $r$ columns contains a 1 (the pivot) as its \emph{lowest} nonzero entry.
\item
Each pivot lies strictly lower than the pivots of the columns to the right of it.
\item
The rightmost $a_{k+1}-r$ rows are entirely zero.
\end{itemize}
These column operations correspond to elementary operations $E_{k+1}(p,q,\lambda)$, and the effect of these operations is felt on the next matrix $M_{k+1}$ or~$N_{k+1}$, which must be modified accordingly.
We now define $\phi_{k+1}$ as follows:
\[
\phi_{k+1}(p) = 
\left\{
\begin{array}{ll}
\phi_k(q)+1 & \qquad \mbox{if column~$p$ has a pivot in row~$q$,}
\\
1 &  \qquad \mbox{if column~$p$ has no pivot.}
\end{array}
\right.
\]
See Figure~\ref{fig:gaussN}.
\begin{figure}
\centerline{
\includegraphics[scale=.5]{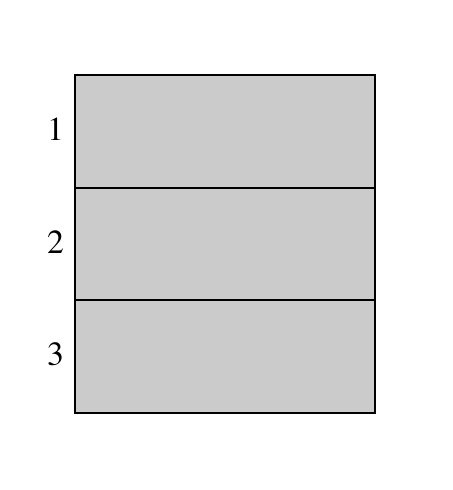}
\qquad
\includegraphics[scale=.5]{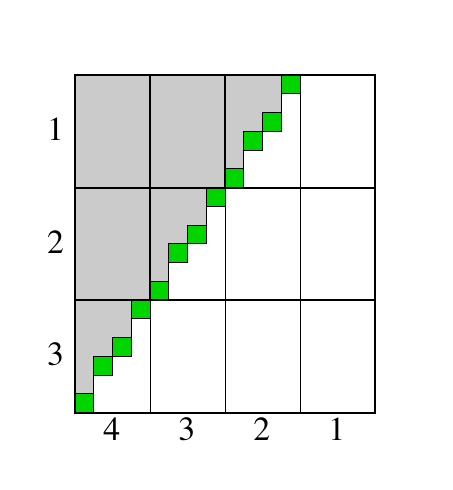}
}
\caption{Using column echelon form to compute $\Rr_{k+1}$ from $\Rr_k$.}
\label{fig:gaussN}
\end{figure}
It is evident in the figure that $R^{k+1}_{i+1}$ is the largest subspace which maps into $R^{k}_i$, for all~$i$.

This concludes our treatment of the concrete form of the zigzag algorithm.

%------------------------------------------------------------------
\section{Further Algebraic Techniques}
\label{sec:further}

%--------------------------------------------------
\subsection{Localization at a single index}
\label{subsec:localize}

Let $\Vv$ be a zigzag module of length~$n$ and let $1 \leq k \leq n$. We consider the problem of determining the set of intervals in  $\Pers(\Vv)$ which contain~$k$, without necessarily computing $\Pers(\Vv)$ itself. We shall see that all the necessary information is contained in a pair of filtrations on the vector space~$V_k$.

\begin{Definition}
Let $\Vv$ be a zigzag module of length~$n$. The \textbf{left-filtration} of $\Vv$ is a filtration on~$V_1$ of depth~$n$, defined as
\[
\lf(\Vv) = \rf(\bar{\Vv})
\]
where $\bar\Vv$ is the reversal of~$\Vv$; so $\bar{V}_i = V_{n+1-i}$, with maps $\bar{f}_i = g_{n-i}$ or $\bar{g}_i = f_{n-i}$.
\end{Definition}

For any~$k$ we therefore have two natural filtrations on~$V_k$:
\[
\begin{array}{rclcl}
\Rr_k &=&  (R_0, R_1, \dots, R_k) &=& \rf(\Vv[1,k]),
\\
\Ll_k &=&  (L_0, L_1, \dots, L_{n+1-k}) &=& \lf(\Vv[k,n]);
\end{array}
\]
the right-filtration over the index set $\{1, \dots, k\}$ and the left-filtration over the index set $\{k, \dots, n\}$.
We also have birth-time and death-time indices
\[
\begin{array}{rclcl}
\bt_k &=& (b_1, \dots, b_k) &=& \bt(\tau[1,k])
\\
\dt_k &=& (d_1, \dots, d_{n+1-k}) &=& n+1 - \bt(\bar\tau[k,n])
\end{array}
\]
which indicate the birth and death times associated with the respective subquotients of $\Rr_k$ and $\Ll_k$. These depend on the type~$\tau$ of~$\Vv$.

\begin{Example}
\label{ex:bifilt}
Consider the zigzag module
\[
\Vv \;=\; (\,
V_1 \fmap{1} V_2 \fmap{2} V_3 \gmap{3} V_4
\,).
\]
At $k=2$, for instance, we have
\[
\begin{array}{rcl}
\Rr_2 &=& (0, f_1(V_1), V_2)
\\
\Ll_2 &=& (0, f_2^{-1}(0), f_2^{-1} g_3 (V_4), V_2)
\end{array}
\]
and
\[
\begin{array}{rcl}
\bt_2 &=& (1, 2)
\\
\dt_2 &=& (2, 4, 3).
\end{array}
\]
\end{Example}

We can now state the main theorem of this section.

\begin{Theorem}[Localization at index~$k$]
\label{thm:kinterval}
Let $\Vv$ be a zigzag module of length~$n$ and let $1 \leq k \leq n$. Let $\Rr_k, \Ll_k$ denote the right- and left-filtrations at~$k$, and let $\bt_k, \dt_k$ denote the birth-time and death-time indices at~$k$.
Then, for all $i,j$ in the range $1 \leq i \leq k$, $1 \leq j \leq n+1-k$, the multiplicity of $[b_i, d_j]$ in $\Pers(\Vv)$ is equal to
\[
c_{ij}
= \dim(R_i \cap L_j)  - \dim(R_{i-1} \cap L_j)
 - \dim(R_i \cap L_{j-1})  + \dim(R_{i-1} \cap L_{j-1}).
%= r_{i,j} - r_{i-1,j} - r_{i, j-1} + r_{i-1, j-1}
\]
\end{Theorem}

\begin{Remark}
Equivalently, $c_{ij} = \dim((R_i \cap L_j) / ((R_{i-1} \cap L_j) + (R_i \cap L_{j-1})))$, the dimension of the $(i,j)$-th bifiltration subquotient.
\end{Remark}

This theorem answers the original question, because every interval containing~$k$ can be written as $[b_i, d_j]$ for some choice of $i,j$. 
We now work towards a proof of Theorem~\ref{thm:kinterval}.

\begin{Proposition}
\label{prop:kstream}
It is sufficient to prove Theorem~\ref{thm:kinterval} in the special case where $\Vv$ is right-streamlined over $\{1, \dots, k\}$ and left-streamlined over $\{k, \dots, n\}$.
\end{Proposition}

\begin{proof}
It is clear from Lemma~\ref{lem:cascade} that we can write $\Vv = \Uu \oplus \Ww$ where $\Uu$ is supported in $\{1, \dots, k-1\}$ and $\Ww$ is right-streamlined over $\{1, \dots, k\}$. Indeed, take $\Uu = \Vv^1 \oplus \dots \oplus \Vv^{k-1}$ and $\Ww = \Vv^k \oplus \dots \oplus \Vv^n$.
Moreover, it is sufficient to prove Theorem~\ref{thm:kinterval} for~$\Ww$, because the filtrations $\Rr_k, \Ll_k$ remain unchanged from~$\Vv$, and the discarded term $\Uu$ decomposes into intervals which do not contain~$k$. Thus, we may assume that $\Vv$ is right-streamlined over $\{ 1, \dots, k\}$.

Repeating this argument from the other side, we may further assume that $\Vv$ is left-streamlined over $\{k, \dots, n\}$. 
\end{proof}

\begin{proof}[Proof of Theorem~\ref{thm:kinterval}]
Assume that $\Vv$ satisfies the condition in Proposition~\ref{prop:kstream}. It follows that every interval in $\Pers(\Vv)$ contains~$k$: any other interval in the decomposition would cause a failure of the streamline condition.
We can therefore write the interval decomposition of~$\Vv$ as
\[
\Vv = \bigoplus_{a \in A} \Ii_a \cong \bigoplus_{a \in A} \Ii(b_{p(a)}, d_{q(a)})
\]
where $A$ indexes the summands, and $p: A \to \{1, \dots, k\}$ and $q: A \to \{1, \dots, n-k+1\}$ identify the interval type of each summand in terms of the birth-time and death-time indices. It is apparent from this formulation that
\[
c_{ij} = \# \{ a \in A \mid p(a) = i,\, q(a) = j\}
\]
and it remains to compute this in terms of the dimensions $\dim(R_i \cap L_j)$.

The interval decomposition restricts at index~$k$ to a direct sum decomposition of $V_k$ into 1-dimensional subspaces~$U_a$, generated by elements $x_a$, say.
Then
\[
\Rr_k
= \bigoplus_{a \in A} \rf(\Ii_a[1,k])
= \bigoplus_{a \in A} \Rr_k \cap U_a
\cong \bigoplus_{a \in A} \Jj(p(a), k)
\]
where the final isomorphism comes from Lemma~\ref{lem:intervalformula}.
Now, the filtration subspace~$R_i$ is spanned by the terms isomorphic to $\Jj(p,k)$ with $p \leq i$. In other words, for $i = 1, \dots, k$ we have
\[
R_i = \Span \left\{ x_a \mid p(a) \leq i \right \}.
\]
A similar argument proceeding from the other direction gives the analogous formula
\[
L_j = \Span \left\{ x_a \mid q(a) \leq j \right \},
\]
for $j = 1, \dots, n+1-k$. Since the $x_a$ are independent, these formulas give bases for $R_i, L_j$.

We now claim that
\[
R_i \cap L_j = \Span \{ x_a \mid p(a) \leq i, \, q(a) \leq j\}
\]
for all $i,j$.
The inclusion $\Span \subseteq R_i \cap L_j$ is obvious, because each of the spanning vectors~$x_a$ belongs to both $R_i$ and~$L_j$. In the other direction, if $x \in R_i \cap L_j$ then write $x = \sum_{a \in A} \lambda_a x_a$. Since $x \in R_i$, all the coefficients $\lambda_a$ with $p(a) > i$ must be zero. Since $x \in L_j$, all the coefficients $\lambda_a$ with $q(a) > j$ must be zero. Thus $x \in \Span \{ x_a \mid p(a) \leq i, \, q(a) \leq j\}$. This establishes the reverse inclusion $R_i \cap L_j \subseteq \Span$ and hence the equality.

Then
\[
\dim(R_i \cap L_j)
= \# \{ x_a \mid p(a) \leq i, \, q(a) \leq j\}
= \sum_{p = 1}^{i} \sum_{q=1}^{j} c_{pq}
\]
for all $i,j$. The formula in the theorem follows easily from this.
\end{proof}

\begin{Remark}
The salient fact behind this result is that it is possible to find a direct sum decomposition of~$V_k$ which simultaneously decomposes the filtered spaces $\Rr_k, \Ll_k$ into intervals within their respective categories $\filt_k$, $\filt_{n+1-k}$.
Here we achieved this by appealing to the interval decomposition of~$\Vv$, but this can also be proved directly for an arbitrary pair of filtrations on a single vector space.
The analogous statement for a triple of filtrations is false. For example
\[
(0, \Ff \oplus 0, \Ff^2),
\qquad
(0, 0 \oplus \Ff, \Ff^2),
\qquad
(0, \Delta, \Ff^2),
\]
(where $\Delta = \{ (x,x) \mid x \in \Ff \}$) cannot be simultaneously decomposed into intervals.
\end{Remark}

%--------------------------------------------------
\subsection{The Diamond Principle}
\label{subsec:diamond}

Consider the following diagram:
\[
\begin{diagram}
\node[4]{W_k}
\\
\node{V_1}
  \arrow{e,t,<>}{p_1}
\node{\cdots}
  \arrow{e,t,<>}{p_{k-2}}
\node{V_{k-1}}
  \arrow{ne,t}{f_{k-1}}
\node[2]{V_{k+1}}
  \arrow{nw,t}{g_{k}}
  \arrow{e,t,<>}{p_{k+1}}
\node{\cdots}
  \arrow{e,t,<>}{p_{n-1}}
\node{V_n}
\\
\node[4]{U_k}
  \arrow{nw,t}{g_{k-1}}
  \arrow{ne,t}{f_{k}}
\end{diagram}
\]
Let $\Vv^+$ and $\Vv^-$ denote the two zigzag modules contained in the diagram:
\begin{eqnarray*}
\Vv^+ &=&
(\, V_1 \lrmap \dots \lrmap V_{k-1}
\fmap{k-1} W_k \gmap{k}
V_{k+1} \lrmap \dots \lrmap V_n
\,)
\\
\Vv^- &=&
(\, V_1 \lrmap \dots \lrmap V_{k-1}
\gmap{k-1} U_k \fmap{k}
V_{k+1} \lrmap \dots \lrmap V_n
\,)
\end{eqnarray*}
We wish to compare $\Pers(\Vv^+)$ with $\Pers(\Vv^-)$, particularly with respect to intervals that meet $\{k-1, k, k+1\}$. This requires a favourable condition on the four maps in the middle diamond.

\begin{Definition}
We say that the diagram
\[
\divide\dgARROWLENGTH by2
\begin{diagram}
\node{V_{k+1}}
  \arrow{e,t}{g_{k}}
\node{W_k}
\\
\node{U_k}
  \arrow{n,l}{f_{k}}
  \arrow{e,t}{g_{k-1}}
\node{V_{k-1}}
  \arrow{n,r}{f_{k-1}}
\end{diagram}
\]
is \textbf{exact} if $\Img(D_1) = \Ker(D_2)$ in the following sequence
\[
\begin{diagram}
\node{U_k}
  \arrow{e,t}{D_1}
\node{V_{k-1} \oplus V_{k+1}}
  \arrow{e,t}{D_2}
\node{W_k}
\end{diagram}
\]
where $D_1(u) = g_{k-1}(u) \oplus f_k(u)$ and $D_2(v \oplus v') = f_{k-1}(v) - g_k(v')$.
\end{Definition}

\begin{Theorem}[The Diamond Principle]
\label{thm:diamond}
Given $\Vv^+$ and $\Vv^-$ as above, suppose that the middle diamond is exact. Then there is a partial bijection of the multisets $\Pers(\Vv^+)$ and $\Pers(\Vv^-)$, with intervals matched according to the following rules:
\begin{itemize}
\item
Intervals of type $[k,k]$ are unmatched.
\item
Type $[b,k]$ is matched with type $[b,k-1]$ and vice versa, for $b \leq k-1$.
\item
Type $[k,d]$ is matched with type $[k+1,d]$ and vice versa, for $d \geq k+1$.
\item
Type $[b,d]$ is matched with type $[b,d]$, in all other cases.
\end{itemize}
%
%See Figures \ref{fig:matched} and~\ref{fig:unmatched}.
%
It follows that the restrictions $\Pers(\Vv^+)|_K$, $\Pers(\Vv^-)|_K$ to the set $K = \{1, \dots, n\} \setminus \{k\}$ are equal.
\end{Theorem}

Figures \ref{fig:diamond_barcode} and~\ref{fig:diamond_pd} illustrate Theorem~\ref{thm:diamond} in terms of barcodes and persistence diagrams, respectively.

\begin{figure}
\centerline{
\includegraphics[scale=0.3]{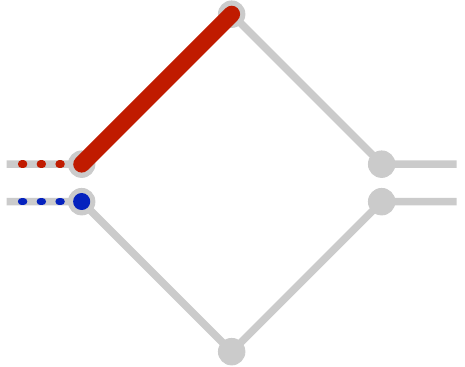}
\hfill
\includegraphics[scale=0.3]{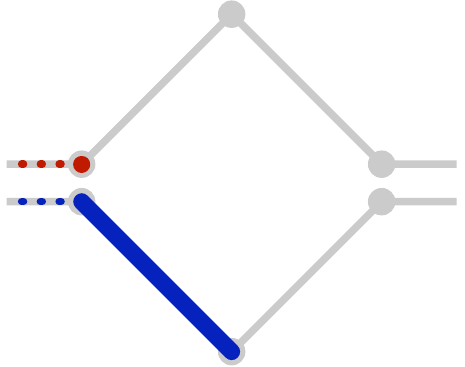}
\hfill
\includegraphics[scale=0.3]{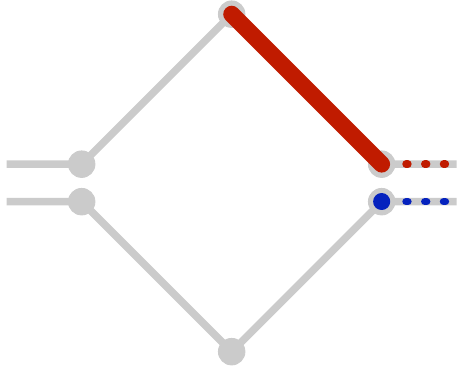}
\hfill
\includegraphics[scale=0.3]{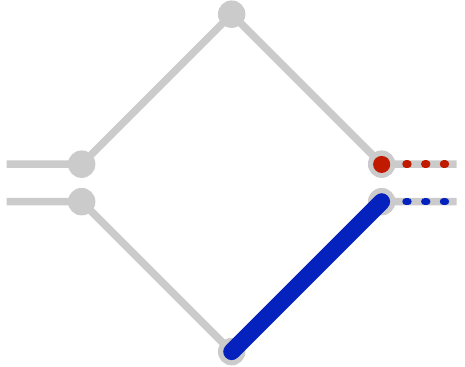}
\hfill
\includegraphics[scale=0.3]{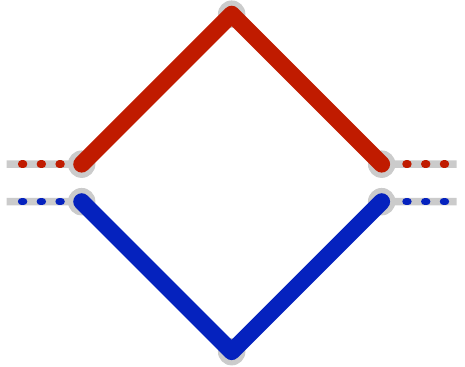}
}
\vspace{2ex}
\centerline{
\hfill\hfill
\includegraphics[scale=0.3]{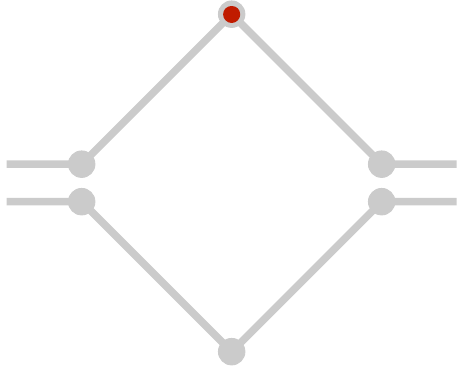}
\hfill
\includegraphics[scale=0.3]{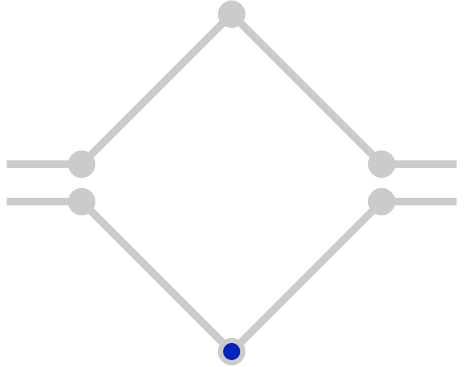}
\hfill\hfill
}
\caption{%
Interval matching between $\Pers(\Vv^+)$ and $\Pers(\Vv^-)$: (top row) the five cases where matching occurs; (bottom row) unmatched intervals $[k,k]$.}
\label{fig:diamond_barcode}
\end{figure}

\begin{figure}
\centerline{
\includegraphics[scale=0.5]{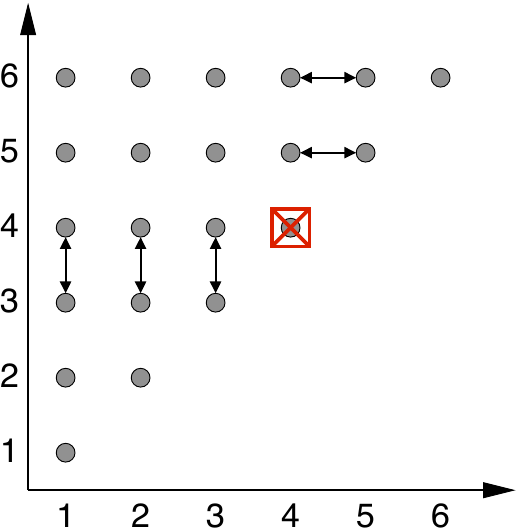}
}
\caption{From $\Pers(\Vv^+)$ to $\Pers(\Vv^-)$, for $n=6$, $k=4$: points in the persistence plane move according to the arrows; the multiplicity of the point marked~$\boxtimes$ changes unpredictably.}
\label{fig:diamond_pd}
\end{figure}

\begin{Remark}
The $\Ii(k,k)$ summands in $\Pers(\Vv^+)$ span the cokernel of~$D_2$, whereas the $\Ii(k,k)$ summands in $\Pers(\Vv^-)$ span the kernel of~$D_1$. The hypothesis of Theorem~\ref{thm:diamond} does not bring about any relation between these spaces (which is why the $[k,k]$ intervals are unmatched).
In Section~\ref{subsec:strongdiamond}, however, we consider a situation in which the $[k,k]$ intervals can be tracked.
\end{Remark}

We use the localization technique of Section~\ref{subsec:localize} to prove Theorem~\ref{thm:diamond}. We begin with birth- and death-time indices.

\begin{Proposition}
\label{prop:diamond}
Let $\tau^+, \tau^-$ denote the zigzag types of~$\Vv^+, \Vv^-$ respectively. If we write
\[
(b_1, \dots, b_{k-1}) = \bt(\tau^+[1,k-1]) = \bt(\tau^-[1,k-1])
\]
for the birth-time index up to time~$k-1$, then
\begin{eqnarray*}
\bt(\tau^+[1,k+1]) &=& (k+1, b_1, \dots, b_{k-1}, k).
\\
\bt(\tau^-[1,k+1]) &=& (k, b_1, \dots, b_{k-1}, k+1),
\end{eqnarray*}
Similarly, if we write
\[
(d_1, \dots, d_{n-k}) = \dt(\tau^+[k+1,n]) = \dt(\tau^-[k+1,n])
\]
for the death-time index from time~$k+1$, then
\begin{eqnarray*}
\dt(\tau^+[k-1,n]) &=& (k-1, d_1, \dots, d_{n-k}, k).
\\
\dt(\tau^-[k-1,n]) &=& (k, d_1, \dots, d_{n-k}, k-1),
\end{eqnarray*}
\end{Proposition}

\begin{proof}
This is immediate from the recursive definition of birth-time index. If we write $\tau_0 = \tau^+[1,k-1] = \tau^-[1,k-1]$ then $\tau^+[1,k+1] = \tau_0 fg$ and $\tau^-[1,k+1] = \tau_0 gf$.
The death-time index is treated similarly.
\end{proof}

Here is the crux of the matter:

\begin{Lemma}
\label{lem:diamond}
In the situation of Theorem~\ref{thm:diamond}, the following filtrations are equal:
\begin{eqnarray*}
\rf(\Vv^+[1,k+1]) &=& \rf(\Vv^-[1,k+1])
\\
\lf(\Vv^+[k-1,n]) &=& \lf(\Vv^-[k-1,n])
\end{eqnarray*}
\end{Lemma}

\begin{proof}
Write $(R_0, R_1, \dots, R_{k-1}) = \rf(\Vv^+[1,k-1]) = \rf(\Vv^-[1,k-1])$.
By the recursive formula (Definition~\ref{def:rfilt}),
\[
\rf(\Vv^+[1,k+1])
=
(0, g_k^{-1} f_{k-1}(R_0), \dots, g_k^{-1} f_{k-1}(R_{k-1}), V_{k+1})
\]
and
\[
\rf(\Vv^-[1,k+1])
=
(0, f_k g_{k-1}^{-1}(R_0), \dots, f_k g_{k-1}^{-1}(R_{k-1}), V_{k+1}).
\]
Thus we can prove the first statement of the lemma by showing that
\[
f_k g_{k-1}^{-1}(R) = g_k^{-1} f_{k-1}(R)
\]
for any subspace $R \leq V_{k-1}$. We use first-order logic. Let $x \in V_{k+1}$. We have the following chain of equivalent statements.
\begin{eqnarray*}
&&
x \in f_k g_{k-1}^{-1}(R)
\\
&\Leftrightarrow&
(\exists z \in R) \, (\exists y \in U_{k}) \, ((g_{k-1}(y) = z)\,\&\,(f_k(y) = x))
\\
&\Leftrightarrow&
(\exists z \in R) \, (\exists y \in U_{k}) \, (D_1(y) = z \oplus x)
\\
&\Leftrightarrow&
(\exists z \in R) \, (z \oplus x \in \Img(D_1))
\end{eqnarray*}
On the other hand:
\begin{eqnarray*}
&&
x \in g_k^{-1} f_{k-1}(R)
\\
&\Leftrightarrow&
(\exists z \in R) \, (f_{k-1}(z) = g_k(x))
\\
&\Leftrightarrow&
(\exists z \in R) \, (z \oplus x \in \Ker(D_2))
\end{eqnarray*}
Since $\Img(D_1) = \Ker(D_2)$ by hypothesis, it follows that $f_k g_{k-1}^{-1}(R) = g_k^{-1} f_{k-1}(R)$.

This proves the first equality. The second equality follows symmetrically.
\end{proof}

\begin{proof}[Proof of Theorem~\ref{thm:diamond}]
We adopt the notation of Section~\ref{subsec:localize}, and consider the right- and left-filtrations at~$V_{k+1}$, for both $\Vv^+$ and $\Vv^-$. Since $\Vv^+[k+1,n] = \Vv^-[k+1,n]$ we have
\[
\Ll_{k+1}^+ = \Ll_{k+1}^-
\qquad
\mbox{and}
\qquad
\dt_{k+1}^+ = \dt_{k+1}^-,
\]
and by Lemma~\ref{lem:diamond} we have
\[
\Rr_{k+1}^+ = \Rr_{k+1}^-.
\]
Finally, $\bt_{k+1}^+$ agrees with $\bt_{k+1}^-$ except that $k, k+1$ are interchanged, according to Proposition~\ref{prop:diamond}.
Thus, when we use Theorem~\ref{thm:kinterval} to calculate the multiplicity of $[b,d]$ for $b \leq k+1 \leq d$, there is perfect agreement between $\Vv^+$ and $\Vv^-$ except that we must interchange $k, k+1$ when they occur as birth-times.

A symmetrical argument can be made, localizing at~$V_{k-1}$. 
%
%Since $\Vv^+[1,k-1] = \Vv^-[1,k-1]$ we have
%$\Rr_{k-1}^+ = \Rr_{k-1}^-$ and $\bt_{k-1}^+ = \bt_{k-1}^-$.
%By Lemma~\ref{lem:diamond}, $\Ll_{k-1}^+ = \Ll_{k-1}^-$. By Proposition~\ref{prop:diamond}, $\dt_{k-1}^+$ agrees with $\dt_{k-1}^-$ except that $k, k-1$ are interchanged.
%
When we compute the multiplicity of $[b,d]$ for $b \leq k-1 \leq d$, there is perfect agreement between $\Vv^+$ and $\Vv^-$ except that we must interchange $k, k-1$ when they occur as death-times.

We have covered all cases of the theorem except for intervals which meet neither $k-1$ nor $k+1$. Intervals contained in $[1,k-2]$ are automatically the same for $\Vv^+$ and $\Vv^-$ because they can be computed by restricting to $\Vv^+[1,k-1]$ and $\Vv^-[1,k-1]$, which are equal. Similarly, intervals contained in $[k+2,n]$ are the same for $\Vv^+$ and $\Vv^-$, by restricting to $\Vv^+[k+1,n] = \Vv^-[k+1,n]$.

Finally, consider intervals $[k,k]$. Nothing can be said about those.
\end{proof}

%--------------------------------------------------
\subsection{The Strong Diamond Principle}
\label{subsec:strongdiamond}

The Diamond Principle can usefully be applied to the following diagram of topological spaces and continuous maps. The four maps in the central diamond are inclusion maps, and the remaining maps $\leftrightarrow$ are arbitrary.
\[
\divide\dgARROWLENGTH by2
\begin{diagram}
\node[5]{A \cup B}
\\
\node{X_1}
  \arrow{e,t,<>}{}
\node{\cdots}
  \arrow{e,t,<>}{}
\node{X_{k-2}}
  \arrow{e,t,<>}{}
\node{A}
  \arrow{ne,t}{}
\node[2]{B}
  \arrow{nw,t}{}
  \arrow{e,t,<>}{}
\node{X_{k+2}}
  \arrow{e,t,<>}{}
\node{\cdots}
  \arrow{e,t,<>}{}
\node{X_n}
\\
\node[5]{A \cap B}
  \arrow{nw,t}{}
  \arrow{ne,t}{}
\end{diagram}
\]
Let $\Xx^+, \Xx^-$ denote the upper and lower zigzag diagrams contained in this picture; so $\Xx^+$ passes through $A \cup B$ and $\Xx^-$, passes through $A \cap B$.

\begin{Theorem}[The Strong Diamond Principle]
\label{thm:strongdiamond}
Given $\Xx^+$ and $\Xx^-$ as above, there is a (complete) bijection between the multisets $\Pers(H_*(\Xx^+))$ and $\Pers(H_*(\Xx^-))$. Intervals are matched according to the following rules:
\begin{itemize}
\item
$[k,k] \in \Pers(H_{\ell+1}(\Xx^+))$ is matched with $[k,k] \in \Pers(H_{\ell}(\Xx^-))$.
\end{itemize}
In the remaining cases, the matching preserves homological dimension:
\begin{itemize}
\item
Type $[b,k]$ is matched with type $[b,k-1]$ and vice versa, for $b \leq k-1$.
\item
Type $[k,d]$ is matched with type $[k+1,d]$ and vice versa, for $d \geq k+1$.
\item
Type $[b,d]$ is matched with type $[b,d]$, in all other cases.
\end{itemize}
\end{Theorem}

\begin{proof}
For any~$\ell$, apply the homology functor $H_\ell$ to the diagram. The central diamond
\[
\divide\dgARROWLENGTH by2
\begin{diagram}
\node{H_\ell(A)}
  \arrow{e}
\node{H_\ell(A \cup B)}
\\
\node{H_\ell(A \cap B)}
  \arrow{n}
  \arrow{e}
\node{H_\ell(B)}
  \arrow{n}
\end{diagram}
\]
is exact by virtue of the Mayer--Vietoris theorem, according to which
\[
\ldots
\longrightarrow
H_\ell(A \cap B)
\stackrel{D_1}{\longrightarrow}
H_\ell(A) \oplus H_\ell(B)
\stackrel{D_2}{\longrightarrow}
H_\ell(A \cup B)
\longrightarrow
\ldots
\]
is an exact sequence.
The Diamond Principle therefore applies to $H_\ell(\Xx^+)$ and $H_\ell(\Xx^-)$, and we have a partial bijection which accounts for all intervals except those of type $[k,k]$.

Now consider the connecting homomorphism in the same Mayer--Vietoris sequence:
\[
\ldots
\stackrel{D_2}{\longrightarrow}
H_{\ell+1}(A \cup B)
\stackrel{\partial}{\longrightarrow}
H_{\ell}(A \cap B)
\stackrel{D_1}{\longrightarrow}
\ldots
\]
By exactness, $\partial$ induces an isomorphism between the cokernel of $D_2$ and the kernel of~$D_1$. But the $[k,k]$ summands of $\Pers(H_{\ell+1}(\Xx^+))$ precisely span $\Coker(D_2)$, whereas the $[k,k]$ summands of $\Pers(H_\ell(\Xx^-))$ span $\Ker(D_1)$.
This establishes the claimed bijection between the $[k,k]$ intervals.
%\[
%\mbox{\rm intervals $[k,k]$ in $\Pers(H_{\ell+1}(\Xx^+))$}
%\;\longleftrightarrow\;
%\mbox{\rm intervals $[k,k]$ in $\Pers(H_{\ell}(\Xx^-))$}
%\]
\end{proof}

\begin{Example}
\label{ex:1param}
Let $\Xx = (X_1, \dots, X_n)$ be a sequence of simplicial complexes defined on a common vertex set. Suppose these have arisen in some context where each transition $X_i$ to $X_{i+1}$ is regarded as being a `small' change. There are two natural zigzag sequences linking the~$X_i$.

The union zigzag, $\Xx_\cup$:
\[
\divide\dgARROWLENGTH by3
\begin{diagram}
\node[2]{X_1 \cup X_2}
\node[2]{\dots}
\node{\dots}
\node[2]{X_{n-1} \cup X_n}
\\
\node{X_1}
  \arrow{ne}
\node[2]{X_2}
  \arrow{nw}
  \arrow{ne}
\node[3]{X_{n-1}}
  \arrow{nw}
  \arrow{ne}
\node[2]{X_{n}}  
  \arrow{nw}
\end{diagram}
\]

The intersection zigzag, $\Xx_\cap$:
\[
\divide\dgARROWLENGTH by3
\begin{diagram}
\node{X_1}
\node[2]{X_2}
\node[3]{X_{n-1}}  
\node[2]{X_{n}}  
\\
\node[2]{X_1 \cap X_2}
  \arrow{nw}
  \arrow{ne}
\node[2]{\dots}
  \arrow{nw}
\node{\dots}
  \arrow{ne}
\node[2]{X_{n-1} \cap X_n}
  \arrow{nw}
  \arrow{ne}
\end{diagram}
\]
We can think of these as being indexed by the half-integers $\{1, 1\half, 2, 2\half, \dots, n\}$.

We can apply the Strong Diamond Principle $n-1$ times to derive the following relationship between the zigzag persistence of the two sequences $\Pers(H_\ell(\Xx_\cap))$ and $\Pers(H_\ell(\Xx_\cup))$.
Restricting to the integer indices, there is a coarse equality:
\[
\Pers(H_\ell(\Xx_\cup))|_{\{1, \dots, n\}}
=
\Pers(H_\ell(\Xx_\cap))|_{\{ 1, \dots, n \}}
\]
More finely, there is a partial bijection between $\Pers(H_\ell(\Xx_\cup))$ and $\Pers(H_\ell(\Xx_\cap))$. Intervals $[k\half, k\half]$ shift homological dimension by $+1$ (from the intersection sequence to the union sequence). Otherwise $[b,d]  \leftrightarrow [b',d']$ where $\{b,b'\}$ is an unordered pair of the form $\{k\half, k+1\}$ and $\{d, d'\}$ is an unordered pair of the form $\{k, k\half\}$; dimension is preserved.
Figure~\ref{fig:N2U} illustrates the complete correspondence as a transformation of the persistence diagram, for $n=5$.
\begin{figure}
\centerline{
\includegraphics[scale=0.5]{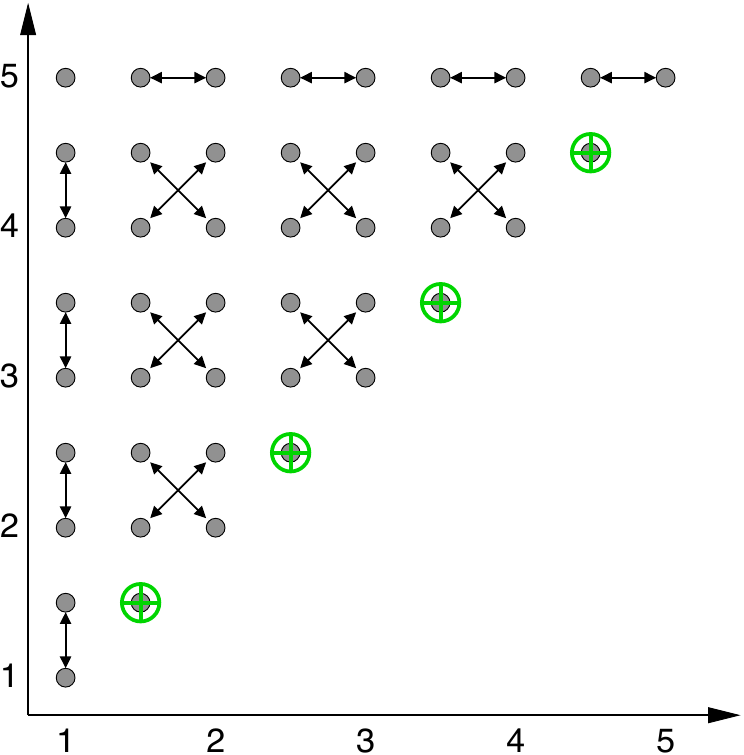}
}
\caption{From $\Pers(H_*(\Xx_\cap))$ to $\Pers(H_*(\Xx_\cup))$, for $n=5$: points in the persistence plane move according to the arrows; points marked $\oplus$ stay fixed and increase homological dimension by~1.}
\label{fig:N2U}
\end{figure}
\end{Example}

%--------------------------------------------------
\section*{Concluding Remarks}
\label{sec:conclusions}

We have presented the foundations of a theory of zigzag persistence which, we believe, considerably extends and enriches the well known and highly successful theory of persistent homology. Zigzag persistence originates in the work of Gabriel and others in the theory of quiver representations. One of our goals has been to bridge the gap between the quiver literature (which is read largely by algebraists) and the current language of applied and computational topology. To this end, we have presented an algorithmic form of Gabriel's structure theorem for $A_n$ quivers, and have indicated the first steps towards integrating these ideas into tools for applied topology.

There are several ways in which this work is incomplete. The most significant omission is an algorithm for computing zigzag persistence in a homological setting (as distinct from the somewhat sanitised vector space algorithm described in Section~\ref{subsec:algconc}). We address this gap in a forthcoming paper with Dmitriy Morozov~\cite{C_dS_M_2008}, where we present an algorithm for computing the zigzag persistence intervals of a 1-parameter family of simplicial complexes on a fixed vertex set.

We have made no effort in this paper to flesh out the applications suggested in Section~\ref{sec:intro}. There is often a substantial gap between the concrete world of point-cloud data sets and the ideal world of simplicial complexes and topological spaces. We intend to develop some of these applications in future work. Meanwhile, we have given priority to establishing the theoretical language and tools. The Diamond Principle is particularly powerful. In the manuscript with Morozov~\cite{C_dS_M_2008}, we show that the Diamond Principle can be used to establish isomorphisms between several different classes of persistence invariants of a space with a real-valued (e.g.\ Morse) function. In particular, we use zigzag persistence to resolve an open conjecture concerning extended persistence~\cite{CS_E_H_2008}. This supports our prejudice that zigzag persistence provides the appropriate level of generality and power for understanding the heuristic concept of persistence in its many manifestations.

%---------------------------------------------
\subsection*{Acknowledgements}
The authors wish to thank Greg Kuperberg, Konstantin Mischaikow and Dmitriy Morozov for helpful conversations and M.~Khovanov for helpful correspondence. The authors gratefully acknowledge support from DARPA, in the form of grants HR0011-05-1-0007 and HR0011-07-1-0002. The second author wishes to thank Pomona College and Stanford University for, respectively, granting and hosting 
his sabbatical during late 2008.

%-------------------------------------------------------
\bibliographystyle{plain}
\bibliography{}

%------------------------------------------------------------------
\end{document}